\pdfoutput=1
\documentclass[11pt,twoside,a4paper,notitlepage]{article}
\usepackage[a4paper,margin=2.5cm]{geometry}

\title{%
Answering \fomodtext queries under updates\\ on bounded degree databases%
\thanks{This is the full version of the conference contribution
  \cite{BKS-ICDT17}.}%
}%

  \author{Christoph Berkholz, Jens Keppeler, Nicole Schweikardt \\
    Humboldt-Universität zu Berlin \\
    \texttt{\{berkholz,keppelej,schweika\}@informatik.hu-berlin.de}}

\usepackage{amsthm}\usepackage{amsmath} \usepackage{amsfonts} \usepackage{amssymb}   

\usepackage[utf8]{inputenc}

\usepackage{tikz}

\usepackage{ifthen}

\usepackage{algorithm}

\usepackage[noend]{algpseudocode}

\usepackage{xspace}

\usepackage{mathabx}

\usepackage{hyperref}

\usepackage{ stmaryrd } %
\usepackage{enumerate}

{}

\makeatletter{}%

       \newtheorem{theorem}{Theorem}[section] 
        \newtheorem{lemma}[theorem]{Lemma}
        \newtheorem{corollary}[theorem]{Corollary}

        \theoremstyle{definition}
        
        \newtheorem{claim}[theorem]{Claim}

\newcommand{\nc}[1]{\newcommand{#1}}
\newcommand{\rnc}[1]{\renewcommand{#1}}

\rnc{\leq}{\ensuremath{\leqslant}}
\rnc{\geq}{\ensuremath{\geqslant}}

\rnc{\le}{\leq}
\rnc{\ge}{\geq}

\nc{\isdef}{\ensuremath{:=}}
\nc{\deff}{\isdef}
\nc{\defi}{\isdef}

\nc{\set}[1]{\ensuremath{\{#1\}}}
\nc{\setsize}[1]{\ensuremath{|#1|}}
\nc{\Setsize}[1]{\ensuremath{\big|#1\big|}}
\nc{\Set}[1]{\ensuremath{\big\{#1\big\}}}
\nc{\setc}[2]{\set{#1 \, : \, #2}}
\nc{\Setc}[2]{\Set{#1 \, : \, #2}}
\nc{\Setcc}[2]{\left\{#1 \, : \, #2\right\}}

\nc{\aufgerundet}[1]{\ensuremath{\lceil #1 \rceil}}
\nc{\abgerundet}[1]{\ensuremath{\lfloor #1 \rfloor}}

\nc{\dcup}{\ensuremath{\dot\cup}}

\nc{\ov}[1]{\ensuremath{\overline{#1}}}

\nc{\NN}{\ensuremath{\mathbb{N}}}
\nc{\NNpos}{\ensuremath{\NN_{\scriptscriptstyle\geq 1}}}
\nc{\RR}{\ensuremath{\mathbb{R}}}
\nc{\RRpos}{\ensuremath{\RR_{\scriptscriptstyle\geq 0}}}

\nc{\fdom}{\ensuremath{\textup{dom}}} %
\nc{\fcodom}{\ensuremath{\textup{codom}}} %
\nc{\und}{\ensuremath{\wedge}}
\nc{\Und}{\ensuremath{\bigwedge}}
\nc{\oder}{\ensuremath{\vee}}
\nc{\Oder}{\ensuremath{\bigvee}}
\nc{\nicht}{\ensuremath{\neg}}
\nc{\impl}{\ensuremath{\to}}
\nc{\gdw}{\ensuremath{\leftrightarrow}}

\nc{\free}{\ensuremath{\textrm{\upshape free}}}
\nc{\qr}{\ensuremath{\textrm{\upshape qr}}}
\nc{\ar}{\ensuremath{\operatorname{ar}}}

\nc{\Structure}[1]{\ensuremath{\mathcal{#1}}}
\nc{\A}{\Structure{A}}
\nc{\B}{\Structure{B}}
\nc{\C}{\Structure{C}}

\nc{\isom}{\ensuremath{\cong}}

\nc{\querycont}{\ensuremath{\sqsubseteq}}
\nc{\eval}[2]{\ensuremath{#1(#2)}}
\nc{\semantik}[1]{\ensuremath{\left\llbracket#1\right\rrbracket}}

\nc{\CanDB}[1]{\ensuremath{\A_{#1}}} %
\nc{\CanTup}[1]{\ensuremath{t_{#1}}} %

\newcommand{\notmodels}{\ensuremath{\nvDash}}

\newcommand{\queryphi}{\varphi}
\newcommand{\varv}{v}
\newcommand{\varu}{u}

\newcommand{\vary}{y}
\newcommand{\varz}{z}

\newcommand{\relS}{S} %
 
\newcommand{\relT}{C} %

\newcommand{\relE}{E} %
\newcommand{\relEtilde}{\tilde{E}} %

\newcommand{\relK}{K} %

\newcommand{\graph}{\mathcal G}
\newcommand{\graphH}{\mathcal H}
\newcommand{\setV}{V} %

\newcommand{\setS}{S}

\newcommand{\equivd}{\equiv_d}

\newcommand{\Odc}{\ensuremath{d^{\bigoh(c^2)}}\xspace}

\newcommand{\GofD}{\graph_{\DB}}

\newcommand{\GofDnew}{\graph_{\DBnew}}
\newcommand{\GofDold}{\graph_{\DBold}}
\newcommand{\graphStart}{\graph_{0}}

\newcommand{\sphere}[2]{\ensuremath{\text{sph}_{#1}(#2)}}

\newcommand{\connsphere}[2]{\ensuremath{\text{conn-sph}_{#1}(#2)}}

\newcommand{\x}{\ov{x}}

\newcommand{\z}{\ov{z}}

\renewcommand{\a}{\ov{a}}
\renewcommand{\b}{\ov{b}}

\nc{\LogicFont}[1]{\ensuremath{\mathsf{#1}}\xspace}
\nc{\FO}{\LogicFont{FO}}
\nc{\FOmod}{\ensuremath{\LogicFont{FO}{+}\LogicFont{MOD}}\xspace}
\nc{\fomod}{\FOmod}
\nc{\fomodtext}{FO{+}MOD\xspace}

\nc{\mymod}{\ensuremath{\textsf{mod}}}
\nc{\existsmc}[2]{\ensuremath{\exists^{#1\textsf{\,mod\,}#2\,}}} %
\nc{\existsim}{\ensuremath{\existsmc{i}{m}}}
\nc{\existsgeq}[1]{\ensuremath{\exists^{\geq #1}}} %
\nc{\existsk}{\ensuremath{\existsgeq{k}}} %
\nc{\existsm}{\ensuremath{\existsgeq{m}}} %
\nc{\existsi}{\ensuremath{\existsgeq{i}}} %

\newcommand{\isomorph}{\ensuremath{\cong}}
\newcommand{\dist}{\ensuremath{\textit{dist}}}

\nc{\Vars}{\ensuremath{\textrm{\upshape vars}}}
\nc{\atoms}{\ensuremath{\textrm{\upshape atoms}}}
\nc{\Adom}{\ensuremath{\textrm{\upshape adom}}}
\nc{\adom}[1]{\ensuremath{\Adom(#1)}} %
\nc{\dom}[1]{\ensuremath{\textrm{\upshape dom}(#1)}} %

\newcommand{\poly}{\operatorname{\textit{poly}}}

\newcommand{\bigoh}{\mathcal{O}}
\newcommand{\bigOh}{\bigoh}
\newcommand{\littleoh}{o}

\newcommand{\parent}{\pointerfont{parent}}

\nc{\arrayfont}[1]{\ensuremath{\texttt{#1}}}

\newcommand{\query}{\ensuremath{\varphi}}

\newcommand{\ccount}{\ensuremath{\mathit{count}}}

\newcommand{\size}[1]{\ensuremath{|\!|#1|\!|}}
\nc{\card}[1]{\ensuremath{|#1|}}

\newcommand{\true}{\ensuremath{\mathtt{true}}\xspace}
\newcommand{\false}{\ensuremath{\mathtt{false}}\xspace}
\nc{\DBold}{\ensuremath{{\DB_{\textit{old}}}}}
\nc{\DBnew}{\ensuremath{{\DB_{\textit{new}}}}}
\newcommand{\UpdateSet}{\ensuremath{{U}}}

\newcommand{\structG}{\ensuremath{\mathcal{G}}}

\newcommand{\indexSetSmall}{\ensuremath{\zeta}}
\newcommand{\tupleSetSmall}[1]{\ensuremath{\mathcal{S}_{#1}}}

\newcommand{\colourG}[1]{\ensuremath{C_{#1}^{\structG}}}

\newcommand{\edgeG}{\ensuremath{E^{\structG}}}

\newcommand{\skipp}[1]{\ensuremath{\operatorname{\textsf{skip}}_{#1}}}

\newcommand{\void}{\ensuremath{\operatorname{\textsf{void}}}}
\newcommand{\firstelement}{\ensuremath{\operatorname{\textsf{first}}}}
\newcommand{\successor}{\ensuremath{\operatorname{\textsf{succ}}}}

\newcommand{\rtdelayd}{\ensuremath{\ell^{3}d}}

\newcommand{\edgeE}[2]{\ensuremath{E_{#1}^{#2}}}

\nc{\arrayE}{\arrayfont{E}}
\nc{\edgeEi}{\ensuremath{E_i}}
\nc{\setU}{\ensuremath{U}}

\newcommand{\potenzmengeof}[1]{2^{#1}}

\newcommand{\ComplexityClass}[1]{\ensuremath{\textsf{\upshape #1}}}
\newcommand{\FPT}{\ComplexityClass{FPT}\xspace}
\newcommand{\AWstar}{\ensuremath{\ComplexityClass{AW}[*]}\xspace}
\newcommand{\AWstarcompl}{\mbox{\ensuremath{\ComplexityClass{AW}[*]}-complete}\xspace}

\nc{\insertp}{\textsc{Insert}}
\nc{\cleanup}{\textsc{cleanUp}}
\nc{\cleanups}{\textsc{cleanUp'}}

\nc{\myparagraph}[1]{\subparagraph*{#1.}}

\nc{\Yes}{\texttt{yes}}
\nc{\No}{\texttt{no}}

\nc{\Dom}{\ensuremath{\textbf{dom}}}
\nc{\Var}{\ensuremath{\textbf{var}}}
\nc{\schema}{\ensuremath{\sigma}}
\nc{\DB}{\ensuremath{D}} %
\nc{\DBStrich}{\ensuremath{D'}} %
\nc{\DBstart}{\ensuremath{{\DB_0}}} %
\nc{\DBempty}{\ensuremath{{\DB_{\emptyset}}}} %

\nc{\DS}{\ensuremath{\mathtt{D}}} %

\rnc{\phi}{\queryphi}

\nc{\UpdateFont}[1]{\ensuremath{\textsf{#1}}}
\nc{\Delete}{\UpdateFont{delete}}
\nc{\Insert}{\UpdateFont{insert}}
\nc{\Update}{\UpdateFont{update}}

\nc{\AlgoFont}[1]{\ensuremath{\textbf{#1}}}
\nc{\PREPROCESS}{\AlgoFont{preprocess}}
\nc{\INIT}{\AlgoFont{init}}
\nc{\UPDATE}{\AlgoFont{update}}
\nc{\ENUMERATE}{\AlgoFont{enumerate}}
\nc{\COUNT}{\AlgoFont{count}}
\nc{\ANSWER}{\AlgoFont{answer}}
\nc{\TEST}{\AlgoFont{test}}

\nc{\EOE}{\texttt{EOE}\xspace} %

\nc{\preprocessingtime}{\ensuremath{t_p}}
\nc{\inittime}{\ensuremath{t_i}}
\nc{\delaytime}{\ensuremath{t_d}}
\nc{\updatetime}{\ensuremath{t_u}}
\nc{\updatetimeStrich}{\ensuremath{t'_u}}

\nc{\answertime}{\ensuremath{t_a}}
\nc{\countingtime}{\ensuremath{t_c}}
\nc{\testingtime}{\ensuremath{t_t}}

\nc{\phiBTypical}{\ensuremath{\phi'_{\relS\text{-}\relE\text{-}\relT}}}
\nc{\phiJTypical}{\ensuremath{\phi_{\relS\text{-}\relE\text{-}\relT}}}
\nc{\phiET}{\ensuremath{\phi_{\relE\text{-}\relT}}}

\nc{\restrict}[2]{\ensuremath{{#1}_{|#2}}}
\nc{\extend}[3]{\ensuremath{{#1}\frac{#3}{#2}}}
\nc{\valuation}{\ensuremath{\beta}}
\nc{\assign}{\ensuremath{\alpha}}
\nc{\assignb}{\ensuremath{\beta}}
\nc{\assignStrich}{\ensuremath{\assign{}'}}
\nc{\emptyassign}{\ensuremath{\emptyset}}
\nc{\Assign}[2]{\ensuremath{\frac{#2}{#1}}}

\nc{\vroot}{\ensuremath{\varv_{\textsl{root}}}}

\nc{\pointerfont}[1]{\textit{#1}}

\nc{\varitem}[1]{\ensuremath{v^{#1}}}
\nc{\assitem}[1]{\ensuremath{\assign^{#1}}}
\nc{\constitem}[1]{\ensuremath{a^{#1}}}
\nc{\parentitem}[1]{\ensuremath{\parent^{#1}}}
\nc{\childitem}[2]{\ensuremath{\pointerfont{child}^{#1}_{#2}}}
\nc{\llist}[2]{\ensuremath{\mathcal{L}_{#2}^{#1}}}
\nc{\nextlistitem}[1]{\ensuremath{\pointerfont{next-listitem}^{#1}}}
\nc{\prevlistitem}[1]{\ensuremath{\pointerfont{prev-listitem}^{#1}}}
\nc{\countitem}[1]{\ensuremath{\ccount^{#1}}}

\nc{\Null}{\ensuremath{0}}

\nc{\arrayA}{\arrayfont{A}}
\nc{\ITEMS}{\mathcal{I}}

\nc{\GG}[1]{\ensuremath{{G}^{#1}}} %
\nc{\GGD}{\GG{D}} 
\nc{\GGT}{\GG{T}} 

\nc{\emptyword}{\ensuremath{\varepsilon}}
\nc{\emptytuple}{\ensuremath{()}}

\nc{\neighb}[3]{\ensuremath{N_{#1}^{#2}(#3)}} %
\nc{\nbset}[2]{\ensuremath{N_{#1}^{#2}}}
\nc{\nrD}[1]{\ensuremath{\neighb{r}{\DB}{#1}}} %
\nc{\nrDStrich}[1]{\ensuremath{\neighb{r}{\DBStrich}{#1}}} %
\nc{\nrT}[1]{\ensuremath{\neighb{r}{T}{#1}}} %
\nc{\Neighb}[3]{\ensuremath{\mathcal{N}_{#1}^{#2}(#3)}} %
\nc{\nb}[2]{\ensuremath{\mathcal{N}_{#1}^{#2}}}

\nc{\NrD}[1]{\ensuremath{\Neighb{r}{\DB}{#1}}} %
\nc{\NrDStrich}[1]{\ensuremath{\Neighb{r}{\DBStrich}{#1}}} %
\nc{\NrT}[1]{\ensuremath{\Neighb{r}{T}{#1}}} %

\nc{\Types}[3]{\ensuremath{\mathcal{T}_{#1}^{\schema,#2}(#3)}} %
\nc{\Typesrdk}{\Types{r}{d}{k}}
\nc{\Typesrd}[1]{\Types{r}{d}{#1}}
\nc{\Typesd}[2]{\Types{#1}{d}{#2}}
\nc{\isotypes}[2]{\ensuremath{\mathcal{T}_{#1}^{#2}}}

\nc{\Typeslist}[3]{\ensuremath{\mathcal{L}_{#1}^{\schema,#2}(#3)}}
\nc{\Typeslistrdk}{\Typeslist{r}{d}{k}}
\nc{\Typeslistrd}[1]{\Typeslist{r}{d}{#1}}
\nc{\Typeslistd}[2]{\Typeslist{#1}{d}{#2}}

\nc{\type}{\ensuremath{\tau}}

\nc{\inducedSubStr}[2]{\ensuremath{#1[#2]}}

\nc{\sgn}[1]{\ensuremath{\textup{sgn}(#1)}}
\nc{\sgnDB}[1]{\ensuremath{\textup{sgn}^{\DB}(#1)}}

\nc{\mycount}{\ensuremath{n}}

\nc{\Ngraph}{\ensuremath{N^{\graph}}}
\nc{\ellEnum}{\ensuremath{s}}

\makeatletter{}%

\begin{document}

\maketitle

\makeatletter{}%

\begin{abstract}
We investigate the query evaluation problem for fixed queries over fully dynamic
databases, where tuples can be inserted or deleted.
The task is to design a dynamic algorithm that
immediately reports the new result of a fixed query after every database update. 

We consider queries in first-order
logic (FO) and its extension with modulo-counting quantifiers
(FO+MOD), and show that they
can be efficiently evaluated under updates, provided that the dynamic database does
not exceed a certain degree bound.

In particular, we construct a data structure that allows to
answer a Boolean FO+MOD query and
to compute the size of the 
result of a non-Boolean query
within constant time after every database update.  Furthermore, after
every update we are able to immediately enumerate the new query result
with constant delay between the output tuples.
The time needed to build the data structure is linear in the size of
the database.

Our results extend earlier work on the evaluation
of first-order queries on static databases of bounded degree and rely
on an effective Hanf normal form for FO+MOD recently obtained by
Heimberg, Kuske, and Schweikardt (LICS 2016).
\end{abstract}

\makeatletter{}%
\section{Introduction}\label{sec:introduction}

Query evaluation is a fundamental task in databases, and a vast amount of
literature is devoted to the complexity of this problem.
In this paper we study query evaluation on relational databases in the
``dynamic setting'', 
where the database may be updated by inserting or deleting tuples.
In this setting, an evaluation
algorithm receives a query $\query$ and an initial database $\DB$ and
starts with a preprocessing phase that computes a suitable data
structure to represent the result of evaluating $\query$ on $\DB$.
After every database update, the data structure is updated so that it
represents the result of evaluating $\query$ on the updated database.
The data structure shall be designed in such a way that it  quickly
provides the query result, preferably in constant time (i.\,e.,
independent of the database size).
We focus on the following flavours of query evaluation. 
\begin{itemize}
\item \emph{Testing:} Decide whether a given tuple $\ov{a}$ is contained in $\query(\DB)$.
\item \emph{Counting:} Compute  $\setsize{\query(\DB)}$ (i.e., the
  number of tuples that belong to $\query(\DB)$).
\item \emph{Enumeration:} Enumerate $\query(\DB)$ with a
  bounded delay between the output tuples.
\end{itemize}
Here, as usual, $\query(\DB)$ denotes the $k$-ary
relation obtained by evaluating a $k$-ary query
$\query$ on a relational database $\DB$.
For \emph{Boolean} queries, all three tasks boil down to
\begin{itemize}
\item \emph{Answering:} Decide if $\phi(\DB)\neq\emptyset$.
\end{itemize}

Compared to the \emph{dynamic descriptive complexity} framework
introduced by Patnaik and Immerman \cite{Patnaik.1997}, which focuses
on the \emph{expressive power} of first-order logic on dynamic databases and
has led to a rich body of literature (see \cite{Schwentick.2016} for
a survey), we are interested in the 
\emph{computational complexity} of query evaluation.
The query language studied in this paper is $\FOmod$,
the extension of first-order logic $\FO$ with modulo-counting
quantifiers of the form  $\existsim x\,\psi$, expressing
that the number of witnesses $x$ that satisfy $\psi$ is congruent to $i$
modulo $m$.
$\FOmod$ can be viewed as a subclass of SQL that properly extends the
relational algebra.

Following \cite{BKS_enumeration_PODS17}, we say that a query evaluation
algorithm is efficient if the update time is either constant or
at most polylogarithmic ($\log^cn$) in the size of the database.
As a consequence, efficient query evaluation in the dynamic setting is only
possible if the static problem (i.e., the setting without database
updates) can be solved in linear or pseudo-linear
($n^{1+\varepsilon}$) time.
Since this is not always possible, we provide a short
overview on known results about first-order query evaluation
on static databases and then proceed by discussing our results in the dynamic setting.

\myparagraph{First-order query evaluation on static databases} 
The problem
of deciding whether a given database $\DB$ satisfies a 
$\FO$-sentence $\query$ is \AWstarcompl (parameterised by $\size{\query}$) and
it is therefore generally believed that the evaluation problem cannot be solved
in time $f(\size{\query})\size{\DB}^{c}$ for any computable $f$ and
constant $c$ (here, $\size{\query}$ and $\size{\DB}$ denote the size
of the query and the database, respectively).  For this reason, a long line of research focused on increasing classes of sparse instances
ranging from databases of \emph{bounded degree} \cite{Seese.1996} (where
every domain element occurs only in a constant number of tuples in the
database) to classes that are \emph{nowhere dense} \cite{Grohe.2014}.
In particular, Boolean first-order queries can be evaluated on
classes of databases of bounded degree in linear time
$f(\size{\query})\size{\DB}$, where the constant factor
$f(\size{\query})$ is 3-fold exponential in $\size{\query}$
\cite{Seese.1996,FrickGrohe_APAL2004}.  As a matter of fact, Frick and Grohe
\cite{FrickGrohe_APAL2004} showed that the 3-fold exponential blow-up
in terms of the query size is unavoidable assuming $\FPT\neq\AWstar$.

Durand and Grandjean \cite{DurandGrandjean_BoundedDegree} and Kazana
and Segoufin \cite{KazanaSegoufin_BoundedDegree} considered the
task of enumerating the result of a $k$-ary first-order query on
bounded degree databases 
and showed that after a linear time preprocessing
phase the query result can be enumerated with constant
delay.
This result was later extended to classes of databases of bounded
expansion \cite{Kazana.2013}.
Kazana and Segoufin \cite{Kazana.2013} also showed that counting the
number of result 
tuples of a $k$-ary first-order query on databases of bounded expansion (and
hence also on databases of bounded degree) can be done in time
$f(\size{\query})\size{\DB}$.
In \cite{DBLP:conf/pods/DurandSS14} an analogous result was obtained
for classes of databases of low degree (i.\,e., degree at most
$\size{\DB}^{o(1)}$) and pseudo-linear time 
$f(\size{\query})\size{\DB}^{1+\varepsilon}$; the paper also presented
an algorithm for enumerating the query results with constant delay
after pseudo-linear time preprocessing.

\myparagraph{Our contribution} We extend the known linear time
algorithms for first-order logic on classes of databases of bounded degree to the
more expressive query language \FOmod. 
Moreover, and more importantly, we lift the tractability to the
dynamic setting 
and show that
the result of $\FO$ and $\FOmod$-queries
can be maintained with constant update
time. In particular, we obtain the following results. Let
$\query$ be a fixed $k$-ary $\FOmod$-query and $d$ a fixed degree bound on the
databases under consideration.  Given an initial database $\DB$, we
construct in linear
time $f(\size{\query},d)\size{\DB}$ a data structure that can be updated in
constant time $f(\size{\query},d)$ when
a tuple is inserted into or
deleted from a relation of $\DB$.
After each update the data structure allows to 
\begin{itemize}
\item immediately answer  $\query$ on $\DB$ if $\query$ is a Boolean
  query (Theorem~\ref{thm:AnsweringBooleanQueries}),
\item test for a given tuple $\ov{a}$ whether $\ov{a}\in\query(\DB)$ in
   time $\bigoh(k^2)$ (Theorem~\ref{thm:testing}),
\item immediately output the number of result tuples $\setsize{\query(\DB)}$ (Theorem~\ref{thm:counting}), and
\item enumerate all tuples $(a_1,\ldots,a_k)\in\query(\DB)$ with
  $\bigoh(k^3)$ delay (Theorem~\ref{thm:enumeration-improved}).
\end{itemize}
For fixed $d$, the parameter function
$f(\size{\query},d)$
is 3-fold exponential
in terms of the query size, which is (by Frick
and Grohe \cite{FrickGrohe_APAL2004}) optimal assuming
$\FPT\neq\AWstar$.

\myparagraph{Outline}
Our dynamic query evaluation algorithm crucially relies on the locality
of \FOmod and in particular an effective Hanf normal form for \FOmod
on databases of bounded degree recently obtained by
Heimberg, Kuske, and Schweikardt \cite{HKS_Hanf_LICS16}.
After some basic definitions in Section~\ref{section:Preliminaries} we
briefly state their result in Section~\ref{section:HNF} and obtain a
dynamic algorithm for Boolean $\FOmod$-queries in Section~\ref{section:BooleanQueries}. 
After some preparations for non-Boolean queries in
Section~\ref{section:TypesAndSpheres}, we present the algorithm for
testing in Section~\ref{section:testing}.
In Section~\ref{section:dbtograph} we reduce the task of counting and
enumerating $\FOmod$-queries in the dynamic setting to the problem of counting
and enumerating independent sets in graphs of bounded degree.
We use this reduction to provide efficient dynamic counting and
enumeration algorithms in Section~\ref{section:counting} and
\ref{section:enumerating}, respectively, and we conclude in Section~\ref{sec:conclusion}.

\myparagraph{Acknowledgements}
Funded by the Deutsche Forschungsgemeinschaft (DFG, German Research
Foundation) -- SCHW~837/5-1.

\makeatletter{}%
\section{Preliminaries}\label{section:Preliminaries}

We write $\NN$ for the set of non-negative integers and let 
$\NNpos\deff\NN\setminus\set{0}$ and $[n]\deff\set{1,\ldots,n}$ for
all $n\in\NNpos$.
By $\potenzmengeof{M}$ we denote the power set of a set $M$.
For a partial function $f$ we write $\fdom(f)$ and $\fcodom(f)$ for the
domain and the codomain of $f$, respectively.

\myparagraph{Databases}
We fix a countably infinite set $\Dom$, the \emph{domain} of potential
database entries. Elements in $\Dom$ are called \emph{constants}.
A \emph{schema} is a finite set $\schema$ of relation symbols, where
each $R\in\schema$ is equipped with a fixed \emph{arity} $\ar(R)\in\NNpos$. 
Let us fix a schema $\schema=\set{R_1,\ldots,R_{|\schema|}}$.
A \emph{database} $\DB$ of schema $\schema$ ($\schema$-db, for short), 
is of the form $\DB=(R_1^\DB,\ldots,R_{|\schema|}^\DB)$, where each 
$R_i^\DB$ is a finite subset of $\Dom^{\ar(R_i)}$.
The \emph{active domain} $\adom{\DB}$ of $\DB$ is the smallest subset
$A$ of $\Dom$ such that $R_i^\DB\subseteq A^{ar(R_i)}$ for each $R_i$
in $\schema$.

The \emph{Gaifman graph} of a $\schema$-db $\DB$ is the undirected
simple graph $\GGD=(V,E)$ with vertex set $V\deff\adom{\DB}$,
where there is an edge between vertices $u$ and $v$ whenever $u\neq v$
and there are
$R\in\schema$ and $(a_1,\ldots,a_{\ar(R)})\in R^{\DB}$ such that 
$u,v\in\set{a_1,\ldots,a_{\ar(R)}}$.
A $\schema$-db $\DB$ is called \emph{connected} if its Gaifman graph
$\GGD$ is connected; the \emph{connected components} of $\DB$ are the
connected components of $\GGD$.
The \emph{degree} of a database $\DB$ is the degree of its Gaifman
graph $\GGD$, i.e., the maximum number of neighbours of a node of $\GGD$.
Throughout this paper we fix a number $d\in\NN$ and restrict
attention to databases of degree at most $d$.

\myparagraph{Updates}
We allow to update a given database of schema $\schema$ by inserting or deleting
tuples as follows (note that both types of commands may change the
database's active domain and the database's degree).
A \emph{deletion} command is of the form
\Delete\,$R(a_1,\ldots,a_r)$
for $R\in\schema$, $r=\ar(R)$, and $a_1,\ldots,a_r\in \Dom$. When
applied to a $\schema$-db $\DB$, it results in the updated $\schema$-db
$\DB'$ with $R^{\DB'}= R^{\DB}\setminus\set{(a_1,\ldots,a_r)}$ and
$S^{\DB'}= S^{\DB}$ for all $S\in\schema\setminus\set{R}$.
\\
An \emph{insertion} command is of the form
\Insert\,$R(a_1,\ldots,a_r)$
for $R\in\schema$, $r=\ar(R)$, and $a_1,\ldots,a_r\in \Dom$. 
When applied to a $\schema$-db $\DB$ in the unrestricted setting, it
results in the updated $\schema$-db 
$\DB'$ with $R^{\DB'}= R^{\DB}\cup\set{(a_1,\ldots,a_r)}$ and
$S^{\DB'}= S^{\DB}$ for all $S\in\schema\setminus\set{R}$.
In this paper, we restrict attention to databases of degree at most
$d$. Therefore, when applying an insertion command to a $\schema$-db $\DB$ of degree
$\leq d$,
the command is carried out only if the resulting
database $\DB'$ still has degree $\leq d$; otherwise $\DB$ remains
unchanged and instead of carrying out the insertion command, an error
message is returned.

\myparagraph{Queries}
We fix a countably infinite set $\Var$ of \emph{variables}.
We consider the extension $\FOmod$ of first-order logic $\FO$ with
modulo-counting quantifiers. For a fixed schema $\schema$, the set
$\FOmod[\schema]$ is built from atomic formulas of the form
$x_1{=}x_2$ and $R(x_1,\ldots,x_{\ar(R)})$, for $R\in\schema$ and
variables $x_1,x_2,\ldots,x_{\ar(R)}\in\Var$, and is closed under Boolean
connectives $\nicht$, $\und$, existential first-order quantifiers
$\exists x$, and modulo-counting quantifiers $\existsim x$, for
a variable $x\in\Var$ and integers $i,m\in\NN$ with $m\geq 2$ and $i<m$.
The intuitive meaning of a formula of the form $\existsim x\,\psi$ is
that the number of witnesses $x$ that satisfy $\psi$ is congruent $i$
modulo $m$.
As usual, $\forall x$, $\oder$, $\impl$, $\gdw$ will be used as
abbreviations when constructing formulas. 
It will be convenient to add the quantifier
$\existsm x$, for $m\in \NNpos$; a formula of
the form $\existsm x\,\psi$ expresses that the number of witnesses
$x$ which satisfy $\psi$ is $\geq m$. This quantifier is just syntactic
sugar an does not increase the expressive power of $\FOmod$.

The \emph{quantifier rank} $\qr(\phi)$ of a $\FOmod$-formula $\phi$ is
the maximum nesting depth of quantifiers that occur in $\phi$.
By $\free(\phi)$ we denote the set of all \emph{free variables} of
$\phi$, i.e., all variables $x$ that have at least one occurrence in
$\phi$ that is not within a quantifier of the form $\exists x$,
$\existsm x$, or
$\existsim x$.
A \emph{sentence} is a formula $\phi$ with $\free(\phi)=\emptyset$.

An \emph{assignment} for $\phi$ in a $\schema$-db $\DB$ is a partial mapping 
$\assign$ from $\Var$ to $\adom{\DB}$, where $\free(\phi)\subseteq
\fdom(\assign)$.
We write $(\DB,\assign)\models\phi$
to indicate that $\phi$ is satisfied when evaluated in $\DB$ with
respect to \emph{active domain semantics} while interpreting every
free occurrence of a variable $x$ with the constant $\assign(x)$.
Recall from \cite{AHV-Book} that ``active domain semantics'' means
that quantifiers are evaluated with respect to the database's active
domain. In particular, $(\DB,\assign)\models \exists x\,\psi$ iff there
exists an $a\in\adom{\DB}$ such that $(\DB,\extend{\alpha}{x}{a})\models \psi$, where
$\extend{\alpha}{x}{a}$ is the assignment $\assignStrich$ with $\assignStrich(x)=a$
and $\assignStrich(y)=\assign(y)$ for all $y\in\fdom(\assign)\setminus\set{x}$.
Accordingly,
$(\DB,\alpha)\models \existsm x\;\psi$ \ iff \
$\big|\setc{\,a\in\adom{\DB}}{(\DB,\alpha{\textstyle
     \frac{a}{x}})\models\psi\,}\big| \; \geq \; m$, \ and \
$(\DB,\alpha)\models \existsim x\;\psi$ \ iff \  
$\big|\setc{\,a\in\adom{\DB}}{(\DB,\alpha{\textstyle
    \frac{a}{x}})\models\psi\,}\big| \; \equiv \; i \textup{ mod } m$\,.

A \emph{$k$-ary $\FOmod$ query of schema $\schema$} is of the form $\phi(x_1,\ldots,x_k)$
where $k\in\NN$, $\phi\in\FOmod[\schema]$, and $\free(\phi)\subseteq
\set{x_1,\ldots,x_k}$.
We will often assume that the tuple $(x_1,\ldots,x_k)$ is clear from
the context and simply write $\phi$ instead of $\phi(x_1,\ldots,x_k)$
and $\big(\DB,(a_1,\ldots,a_k)\big)\models\phi$
instead of
$\big(\DB,\textstyle{\frac{a_1,\ldots,a_k}{x_1,\ldots,x_k}}\big)\models\phi$, 
where $\Assign{x_1,\ldots,x_k}{a_1,\ldots,a_k}$ denotes the assignment
$\assign$ with $\assign(x_i)=a_i$ for all $i\in [k]$.
When evaluated in a $\schema$-db $\DB$,
the $k$-ary query $\phi(x_1,\ldots,x_k)$ yields the $k$-ary relation 
\[
  \phi(\DB)
  \quad\deff\quad
  \big\{\,
    (a_1,\ldots,a_k)\,\in\,\adom{\DB}^k \ : \ 
    \big(\DB,\textstyle{\frac{a_1,\ldots,a_k}{x_1,\ldots,x_k}}\big)\
    \models\ \phi\,
  \big\}\,.
\] 

\emph{Boolean} queries are $k$-ary queries with $k=0$. As usual, for
Boolean queries we will write 
$\phi(\DB)=\No$ instead of $\phi(\DB)=\emptyset$, and 
$\phi(\DB)=\Yes$ instead of $\phi(\DB)\neq\emptyset$; and we write 
$\DB\models \phi$ to indicate that $(\DB,\assign)\models\phi$ for any
assignment $\alpha$.

\myparagraph{Sizes and Cardinalities}
The \emph{size} $\size{\schema}$ of a schema $\schema$ is 
the sum of the arities of its relation symbols.
The size $\size{\query}$ of an $\FOmod$ query $\query$ of schema $\schema$ is 
the length of $\query$ when viewed as a word over the alphabet 
$\schema\cup\Var\cup\NN\cup\set{\,=,\und,\nicht,\exists,{}^{\mymod},{}^{\geq},(,)\,}$.
For a $k$-ary query $\query(x_1,\ldots,x_k)$ and a $\sigma$-db $\DB$, the 
\emph{cardinality of the query result} is the number $|\query(\DB)|$ of
tuples in $\query(\DB)$.
The \emph{cardinality} $\card{\DB}$ of a $\schema$-db $\DB$ is defined
as the number of tuples stored in $\DB$, i.e.,
$\card{\DB}\deff\sum_{R\in\schema} |R^{\DB}|$. 
The \emph{size} $\size{\DB}$ of $\DB$ is defined as
$\size{\schema}+|\Adom(\DB)|+\sum_{R\in\schema} \ar(R){\cdot}
|R^D|$ and corresponds to the size of a reasonable encoding of $\DB$.

\myparagraph{Dynamic Algorithms for Query Evaluation}
We adopt the framework for dynamic algorithms for query evaluation of 
\cite{BKS_enumeration_PODS17}; the next paragraphs are taken almost
verbatim from \cite{BKS_enumeration_PODS17}.
Following \cite{Cormen.2009}, we use Random Access Machines (RAMs)
with $\bigoh(\log n)$ word-size and a uniform cost 
measure to analyse our algorithms.
We will assume that the RAM's memory is initialised to $\Null$. In
particular, if an algorithm uses an array, we will assume
that all array entries are initialised to $\Null$, and this initialisation
comes at no cost (in real-world computers this can be achieved by using the
\emph{lazy array initialisation technique}, cf.\ e.g.\ \cite{MoretShapiro}). 
A further assumption is that for every fixed
dimension $k\in\NNpos$ we have available an unbounded number of
$k$-ary arrays $\arrayA$ such that for given $(n_1,\ldots,n_k)\in\NN^k$
the entry $\arrayA[n_1,\ldots,n_k]$ 
at position $(n_1,\ldots,n_k)$ can be accessed in constant
time.\footnote{While this can be accomplished easily in the RAM-model,  
for an implementation on real-world
computers one would probably have to resort to replacing our use of
arrays by using suitably designed hash functions.}
For our purposes it will be convenient to assume that $\Dom=\NNpos$.

Our algorithms will take as input a $k$-ary $\FOmod$-query
$\query(x_1,\ldots,x_k)$, a parameter
$d$, and a $\schema$-db $\DBstart$ of degree $\leq d$.
For all query evaluation problems considered in this paper, we aim at
routines $\PREPROCESS$ and $\UPDATE$ which achieve the following.
 
Upon input of $\query(x_1,\ldots,x_k)$ and $\DBstart$, $\PREPROCESS$ builds a data
structure $\DS$ which represents $\DBstart$ (and which is designed in
such a way that it supports the evaluation of $\query$ on $\DBstart$).
Upon input of a command $\Update\ R(a_1,\ldots,a_r)$ (with
$\Update\in\set{\Insert,\Delete}$), 
calling $\UPDATE$  modifies the data structure $\DS$ such that it
represents the updated database $\DB$.
The \emph{preprocessing time} $\preprocessingtime$ is the
time used for performing $\PREPROCESS$;
the \emph{update time} $\updatetime$ is the time used for performing
an $\UPDATE$. In this paper, $\updatetime$ will be independent of the size
of the current database $\DB$.
By $\INIT$ we denote the particular case of the routine $\PREPROCESS$
upon input of a query $\query(x_1,\ldots,x_k)$ and the \emph{empty} database
$\DBempty$ (where $R^{\DBempty}=\emptyset$ for all $R\in\schema$).
The \emph{initialisation time} $\inittime$
is the time used for performing $\INIT$.
In all dynamic algorithms presented in this paper, the $\PREPROCESS$ routine
for input of $\query(x_1,\ldots,x_k)$ and $\DBstart$ 
will carry out the $\INIT$ routine for $\query(x_1,\ldots,x_k)$ and then
perform a sequence of $\card{\DBstart}$ update operations to
insert all the tuples of $\DBstart$ into the data structure.
Consequently, $\preprocessingtime = \inittime +  \card{\DBstart}\cdot\updatetime$.

In the following, $\DB$ will always denote the database that is
currently represented by the data structure $\DS$.

To solve the \emph{enumeration problem under updates}, 
apart from the routines $\PREPROCESS$ and $\UPDATE$,
we aim at a routine $\ENUMERATE$ such that
calling $\ENUMERATE$ invokes an enumeration of all tuples
(without repetition) that belong to the query result $\query(\DB)$.
The \emph{delay} $\delaytime$ is the maximum time
used during a call of $\ENUMERATE$
\begin{itemize}
\item until the output of the first tuple (or the end-of-enumeration
  message $\EOE$, 
  if $\phi(\DB)=\emptyset$),
\item between the output of two consecutive tuples, and 
\item between the output of the last tuple and the end-of-enumeration
  message $\EOE$.
\end{itemize}

To \emph{test} if a given tuple belongs to the query result,
instead of $\ENUMERATE$ we aim at a routine $\TEST$ which
upon input of a tuple $\ov{a}\in\Dom^k$ checks whether $\ov{a}\in
\query(\DB)$.
The \emph{testing time} $\testingtime$ is the time used for
performing a $\TEST$.
To solve the \emph{counting problem under updates}, instead of
$\ENUMERATE$ or $\TEST$ we aim at a routine $\COUNT$ which outputs the cardinality
$|\query(\DB)|$ of the query result.
The \emph{counting time} $\countingtime$ is the time used for
performing a $\COUNT$.
To \emph{answer} a \emph{Boolean} query under updates,
instead of $\ENUMERATE$, $\TEST$, or $\COUNT$ we aim at a routine $\ANSWER$
that produces the answer $\Yes$ or $\No$ of $\query$ on $\DB$.
The \emph{answer time} $\answertime$ is the time used for
performing $\ANSWER$.
Whenever speaking of a \emph{dynamic algorithm}, we mean an algorithm
that has routines $\PREPROCESS$ and $\UPDATE$ and, depending on the
problem at hand, at least one of the routines 
$\ANSWER$,
$\TEST$, 
$\COUNT$,
and $\ENUMERATE$.

Throughout the paper, we often adopt the view of \emph{data complexity} 
and suppress factors that may depend on the query $\query$ or the
degree bound $d$, but not on the database $\DB$. 
E.g., ``linear preprocessing time'' means 
$\preprocessingtime\leq f(\query,d)\cdot\size{\DBstart}$ and 
``constant update time'' means  $\updatetime\leq f(\query,d)$, for a
function $f$ with codomain $\NN$. 
When writing $\poly(n)$ we mean $n^{\bigOh(1)}$.

\makeatletter{}%

\section{Hanf Normal Form for \fomodtext}\label{section:HNF}

Our algorithms for evaluating $\FOmod$ queries rely on a decomposition
of $\FOmod$ queries into \emph{Hanf normal form}. To describe this
normal form, we need some more notation.

Two formulas $\phi$ and $\psi$ of schema $\schema$ are called
\emph{$d$-equivalent} (in symbols: $\phi\equivd\psi$) if
for all $\schema$-dbs $\DB$ \emph{of degree $\leq d$} and all
assignments $\alpha$ for $\phi$ and $\psi$ in $\DB$ we have
\ $(D,\alpha)\models\phi$ $\iff$ $(D,\alpha)\models \psi$.

For a $\schema$-db $\DB$ and a set $A\subseteq\adom{\DB}$ we
write $\inducedSubStr{\DB}{A}$ to denote the restriction of $\DB$ to
the domain $A$, i.e.,
$R^{\inducedSubStr{\DB}{A}}=\setc{\ov{a}\in R^\DB}{\ov{a}\in
  A^{\ar(R)}}$, for all $R\in\schema$.
For two $\schema$-dbs $\DB$ and $\DBStrich$ and two $k$-tuples
$\ov{a}=(a_1,\ldots,a_k)$ and $\ov{a}'=(a'_1,\ldots,a'_k)$ of elements
in $\adom{\DB}$ and $\adom{\DBStrich}$, resp., we write
\,$\big(\DB,\ov{a}\big) \isom \big(\DBStrich,\ov{a}'\big)$\,
to indicate that there is an
isomorphism\footnote{An \emph{isomorphism} $\pi\colon \DB\to\DBStrich$
is a bijection from $\adom{\DB}$ to $\adom{\DBStrich}$ 
with $(b_1,\ldots,b_r)\in R^{\DB} \iff (\pi(b_1),\ldots,\pi(b_r))\in R^{\DBStrich}$
for all $R\in\schema$, for $r\deff\ar(R)$, 
and for all $b_1,\ldots,b_r\in\adom{\DB}$.}
$\pi$ from $\DB$ to $\DBStrich$ that maps $a_i$ to
$a'_i$ for all $i\in[k]$.

The \emph{distance} $\dist^\DB(a,b)$ between two elements
$a,b\in\adom{\DB}$ is the minimal length (i.e., the number of edges)
of a path from $a$ to $b$ in $\DB$'s Gaifman graph $\GGD$ (if no such path exists, we let
$\dist^\DB(a,b)=\infty$; note that $\dist^\DB(a,a)=0$).
For $r\geq 0$ and $a\in\Adom(\DB)$, the \emph{$r$-ball} around
$a$ in $\DB$ is the set
$\nrD{a}\deff\setc{b\in\Adom(\DB)}{\dist^{\DB}(a,b)\leq r}$.
For a $\schema$-db $\DB$ and a tuple $\ov{a}=(a_1,\ldots,a_k)$ we let $\nrD{\ov{a}}\deff
\bigcup_{i\in[k]} \nrD{a_i}$.
The \emph{$r$-neighbourhood} around $\ov{a}$ in $\DB$ is defined as the
$\schema$-db $\NrD{\ov{a}}\deff \inducedSubStr{\DB}{\nrD{\ov{a}}}$.

For $r\geq 0$ and $k\geq 1$, a \emph{type $\type$ (over $\schema$)
  with $k$ centres and radius $r$} (for short: \emph{$r$-type with $k$ centres}) is of the form
$(T,\ov{t})$, where $T$ is a $\schema$-db,
$\ov{t}\in\adom{T}^k$, and $\adom{T}=\nrT{\ov{t}}$. The elements in
$\ov{t}$ are called the \emph{centres} of $\type$.
For a tuple $\ov{a}\in\adom{\DB}^k$, the \emph{$r$-type of $\ov{a}$
  in $\DB$} is defined as the $r$-type with $k$ centres
$\big(\NrD{\ov{a}},\ov{a}\big)$. 

For a given $r$-type with $k$ centres $\tau=(T,\ov{t})$ it
is straightforward to construct a first-order formula
$\sphere{\tau}{\ov{x}}$ (depending on $r$ and $\tau$)
with $k$ free variables $\ov{x}=(x_1,\ldots,x_k)$ 
which expresses that the $r$-type of $\ov{x}$ is isomorphic to
$\tau$, i.e., for every $\schema$-db $\DB$
and all $\ov{a}=(a_1,\ldots,a_k)\in\adom{\DB}^k$ we have 
\ $\big(\DB,\ov{a}\big) \models  \sphere{\tau}{\ov{x}}
   \iff 
   \big(\NrD{\ov{a}},\ov{a}\big) \isom \big(T,\ov{t}\big)$.
The formula $\sphere{\tau}{\ov{x}}$ is called a \emph{sphere-formula}
(over $\schema$ and $\ov{x}$);
the numbers $r$ and $k$ are called \emph{locality radius} and
\emph{arity}, resp., of the sphere-formula. 

A \emph{Hanf-sentence} (over $\schema$) is a sentence of the form
\,$\existsm x\; \sphere{\tau}{x}$\, or \,$\existsim x\;
\sphere{\tau}{x}$,\, where $\tau$ is an
$r$-type (over $\schema$) with 1 centre, for some $r\geq 0$. 
The number $r$ is called  \emph{locality radius} of the Hanf-sentence.
A formula in \emph{Hanf normal form} (over $\schema$) is a Boolean
combination\footnote{Throughout this paper, whenever we speak of
  \emph{Boolean combinations} we mean \emph{finite} Boolean
  combinations.} of
sphere-formulas and Hanf-sentences (over $\schema$).
The \emph{locality radius} of a formula $\psi$ in Hanf normal form
is the maximum of the locality radii of the Hanf-sentences and
the sphere-formulas that occur in $\psi$.
The formula is \emph{$d$-bounded} if all types $\tau$
that occur in sphere-formulas or Hanf-sentences of $\psi$ are
$d$-bounded, i.e., $T$ is of degree $\leq d$, where $\tau=(T,\ov{t})$.
Our query evaluation algorithms for $\FOmod$ rely on the following
result by Heimberg, Kuske, and Schweikardt \cite{HKS_Hanf_LICS16}.

\begin{theorem}[\cite{HKS_Hanf_LICS16}]\label{thm:HNF}
There is an algorithm which receives as input a degree bound $d\in\NN$
and a $\FOmod[\schema]$-formula $\phi$, and constructs a
$d$-equivalent formula $\psi$ in Hanf normal form (over $\schema$)
with the same free variables as $\phi$.
For any $d\geq 2$, the formula $\psi$ is $d$-bounded and has locality
radius $\leq 4^{\qr(\phi)}$,
and the algorithm's runtime is
$2^{d^{2^{\bigOh(\size{\phi}+\size{\schema})}}}$. 
\end{theorem}

The first step of all our query evaluation algorithms is to use
Theorem~\ref{thm:HNF} to transform
a given query $\phi(\ov{x})$ into a $d$-equivalent query
$\psi(\ov{x})$ in Hanf normal form.
The following lemma summarises easy facts that are useful for
evaluating the sphere-formulas that occur in $\psi$.

\begin{lemma}\label{lem:basic_facts}
  Let $d\geq 2$ and
  let $\DB$ be a $\sigma$-db of degree $\leq d$.
  Let $r\geq 0$, $k\geq 1$, and $\ov{a}=(a_1,\ldots,a_k)\in\adom{\DB}$.
  \begin{enumerate}[(a)]
  \item\label{eq:Nsize:lem:basic_facts}
    $\Setsize{N_{r}^{\DB}(\ov{a})} \ \leq \
    k\sum^r_{i=0}d^i \ \leq \ kd^{r+1}$.
  \item\label{eq:Ncompute:lem:basic_facts}
    Given $\DB$ and $\ov{a}$, the $r$-neighbourhood $\NrD{\ov{a}}$ can
    be computed in time 
    $\bigl(k d^{r+1}\bigr)^{\bigoh(\size{\sigma})}$.
  \item\label{eq:Nconn2:lem:basic_facts}
    $\NrD{a_1,a_2}$ is connected if and only if
    $\dist^\DB(a_1,a_2)\leq 2r+1$.
  \item\label{eq:Nconnk:lem:basic_facts}
    If $\NrD{\ov{a}}$ is connected, then
    $\nrD{\ov{a}}\subseteq \neighb{r+(k-1)(2r+1)}{\DB}{a_i}$,
    for all $i\in[k]$.
  \item\label{eq:Nisom:lem:basic_facts}
    Let $\DBStrich$ be a $\schema$-db of degree $\leq d$ and let
    $\ov{b}=(b_1,\ldots,b_k)\in\adom{\DBStrich}$. 

    It can be tested in time
    \ $(k d^{r+1})^{\bigOh(\size{\sigma}+ k d^{r+1})} \ \leq \ 
       2^{\bigOh(\size{\schema}k^2d^{2r+2})}$ \
    whether \\
    $\big(\NrD{\ov{a}},\ov{a}\big) \ \isomorph\
    \big(\NrDStrich{\ov{b}},\ov{b}\big)$.
  \end{enumerate}
\end{lemma}

\begin{proof}
Parts \eqref{eq:Nsize:lem:basic_facts}--\eqref{eq:Nconnk:lem:basic_facts} are straightforward.
Concerning Part~\eqref{eq:Nisom:lem:basic_facts}, a brute-force
approach is to loop through all mappings from 
$\nrD{\ov{a}}$ to $\nrDStrich{\ov{b}}$ that map $a_i$ to $b_i$ for
every $i\in [k]$ and check whether this 
mapping is an isomorphism. 
Each such check can be accomplished in time $n^{\bigOh(\size{\sigma})}$ for
$n\deff k d^{r+1}$,
and the number of mappings that have to be checked is $\leq
n^n$. Thus, the isomorphism test is accomplished in time 
$n^{\bigOh(n+\size{\sigma})} = (k d^{r+1})^{\bigOh(\size{\sigma}+ k d^{r+1})}$.
\end{proof}

The time bound stated in 
part~\eqref{eq:Nisom:lem:basic_facts} of Lemma~\ref{lem:basic_facts}
is obtained by a brute-force approach. When using Luks' 
polynomial time isomorphism test for bounded degree
graphs \cite{DBLP:journals/jcss/Luks82}, the time bound of
Lemma~\ref{lem:basic_facts}\eqref{eq:Nisom:lem:basic_facts} can be
improved to  $\bigl(k d^{r+1}\bigr)^{\poly(d\size{\sigma})}$.
However, the asymptotic overall runtime of our 
algorithms for evaluating $\FOmod$-queries won't
improve when using Luks algorithm instead of the brute-force isomorphism
test of Lemma~\ref{lem:basic_facts}\eqref{eq:Nisom:lem:basic_facts}.

\makeatletter{}%

\section{Answering Boolean \fomodtext Queries Under Updates}\label{section:BooleanQueries}

In \cite{FrickGrohe_APAL2004}, Frick and Grohe showed that in the
static setting (i.e., without database updates), Boolean $\FO$-queries $\phi$
can be answered on databases $\DB$ of degree $\leq d$ in time 
$2^{d^{2^{\bigOh(\size{\phi})}}}{\cdot} \size{\DB}$.
Our first main theorem extends their result to $\FOmod$-queries and
the dynamic setting.

\begin{theorem}\label{thm:AnsweringBooleanQueries}
There is a dynamic algorithm that receives a schema $\schema$, a
degree bound $d\geq 2$, a Boolean $\FOmod[\schema]$-query $\phi$, and a
$\schema$-db $\DBstart$ of degree $\leq d$, and computes within 
$\preprocessingtime= f(\phi,d)\cdot\size{\DBstart}$
preprocessing time a data structure that can be updated in time
$\updatetime= f(\phi,d)$ and allows to
return the query result $\phi(\DB)$ with answer time
$\answertime= \bigOh(1)$.
The function $f(\phi,d)$ is of the form 
$2^{d^{2^{\bigOh(\size{\phi})}}}$.

If $\phi$ is a $d$-bounded Hanf-sentence of locality radius $r$,
then $f(\phi,d)= 2^{\bigOh(\size{\schema} d^{2r+2})}$,
and the initialisation
time is $\inittime=\bigOh(\size{\phi})$.
\end{theorem}
\begin{proof}
W.l.o.g.\ we assume that all the symbols of $\schema$ occur in $\phi$
(otherwise, we remove from $\schema$ all symbols that do not occur in
$\phi$).
In the preprocessing routine, we first use Theorem~\ref{thm:HNF} to
transform $\phi$ into a $d$-equivalent sentence $\psi$ in Hanf normal
form; this takes time $2^{d^{2^{\bigOh(\size{\phi})}}}$.
The sentence $\psi$ is a Boolean combination of $d$-bounded
Hanf-sentences (over $\schema$) of locality radius at most $r\deff
4^{\qr(\phi)}$. 
Let $\rho_1,\ldots,\rho_s$ be the list of all types that occur in
$\psi$. Thus, every Hanf-sentence in $\psi$ is of the form
$\existsk x\; \sphere{\rho_j}{x}$ or 
$\existsim x\; \sphere{\rho_j}{x}$ for some $j\in[s]$ and
$k,i,m\in\NN$ with $k\geq 1$, $m\geq 2$, and $i<m$.
For each $j\in[s]$ let $r_j$ be the radius of 
$\sphere{\rho_j}{x}$.
Thus, $\rho_j$ is an $r_j$-type with 1 centre (over $\schema$).

For each $j\in[s]$ our data structure will store the number
$\arrayA[j]$ of all elements $a\in\adom{\DB}$ whose $r_j$-type is
isomorphic to $\rho_j$, i.e., $(\Neighb{r_j}{\DB}{a},a) \isom \rho_j$.
The initialisation for the empty database $\DBempty$ lets
$\arrayA[j]= 0$ for all $j\in[s]$.
In addition to the array $\arrayA$, our data structure
stores a Boolean value $\texttt{Ans}$ where $\texttt{Ans}=\phi(\DB)$
is the answer of the Boolean query $\phi$ on the current database
$\DB$. This way, the query can be answered in time
$\bigOh(1)$ by simply outputting $\texttt{Ans}$.
The initialisation for the empty database $\DBempty$ computes
$\texttt{Ans}$ as follows. Every Hanf-sentence of the form $\existsk
x\,\sphere{\rho_j}{x}$ in $\psi$ is replaced by the Boolean constant $\false$.
Every Hanf-sentence of the form $\existsim x\,\sphere{\rho_j}{x}$ is replaced by
$\true$ if $i=0$ and by $\false$ otherwise. The resulting formula, a
Boolean combination of the Boolean constants $\true$ and $\false$,
then is evaluated, and we let $\texttt{Ans}$ be the obtained result.
The entire initialisation takes time at most 
$\inittime= f(\phi,d)= 2^{d^{2^{\bigOh(\size{\phi})}}}$. If $\phi$ is
a Hanf-sentence, we even have $\inittime=\bigOh(\size{\phi})$.

To update our data structure upon a command
$\Update\,R(a_1,\ldots,a_k)$, for $k=\ar(R)$ and 
$\Update\in\set{\Insert,\Delete}$, we proceed as follows.
The idea is to remove from the data structure the
information on all the database elements  
whose $r_j$-neighbourhood (for some $j\in[s]$) is affected by
the update, and then to recompute
the information concerning all these elements
 on the updated database.

Let $\DBold$ be the database before the update is received and
let $\DBnew$ be the database after the update has been performed. 
We consider each $j\in[s]$.
All elements whose $r_j$-neighbourhood might have changed, belong to
the set 
\;$\UpdateSet_j\deff\neighb{r_j}{\DBStrich}{\ov a}$, \ where
$\DBStrich\deff\DBnew$ if the update command is $\Insert\; R(\ov{a})$,
and $\DBStrich\deff\DBold$ if the update command is $\Delete\; R(\ov{a})$.

To remove the old information from $\arrayA[j]$, we compute for each
$a \in \UpdateSet_j$ the neighbourhood
$T_a\deff \Neighb{r_j}{\DBold}{a}$, check whether
$(T_a,a)\isom\rho_j$, and if so, 
decrement the value $\arrayA[j]$. 
\\
To recompute the new information for $\arrayA[j]$, we compute for all 
$a \in \UpdateSet_j$ the neighbourhood
$T'_a\deff \Neighb{r_j}{\DBnew}{a}$, check whether
$(T'_a,a)\isom\rho_j$, and if so, 
increment the value $\arrayA[j]$. 

Using Lemma~\ref{lem:basic_facts} we obtain for each $j\in[s]$ that
$|\UpdateSet_j|\leq kd^{r_j+1}$.
For each $a\in \UpdateSet_j$, the
neighbourhoods $T_a$ and $T'_a$ can be computed in time
$\big(d^{r_j+1}\big)^{\bigOh(\size{\schema})}$, and testing for
isomorphism with $\rho_j$ can be done in time 
$\big(d^{r_j+1}\big)^{\bigOh(\size{\schema}+d^{r_j+1})}$. 
Thus, the update of $\arrayA[j]$ is done in time
$k{\cdot}\big(d^{r_j+1}\big)^{\bigOh(\size{\schema}+d^{r_j+1})}
 \leq
 2^{d^{2^{\bigOh(\size{\phi})}}}$ 
(note that $k\leq \size{\schema}\leq\size{\query}$ and 
$r_j\leq 4^{\qr(\query)}\leq 2^{\bigOh(\size{\query}))}$).

After having updated $\arrayA[j]$ for each $j\in[s]$, we recompute the 
query answer $\texttt{Ans}$ as follows.
Every Hanf-sentence of the form 
$\existsk x\,\sphere{\rho_j}{x}$ in $\psi$ is replaced by the Boolean constant $\true$ if 
$\arrayA[j]\geq k$, and by the Boolean constant $\false$ otherwise.
Every Hanf-sentence of the form $\existsim x\,\sphere{\rho_j}{x}$ is replaced by
$\true$ if $\arrayA[j]\equiv i\ \mymod\ m$, and by $\false$ otherwise. 
The resulting formula, a Boolean combination of the Boolean constants
$\true$ and $\false$, then is
evaluated, and we let $\texttt{Ans}$ be the obtained result.
Thus, recomputing $\texttt{Ans}$ takes time $\poly(\size{\psi})$.

In summary, the entire update time is 
$\updatetime=f(\phi,d)=2^{d^{2^{\bigOh(\size{\phi})}}}$.
In case that $\phi$ is a $d$-bounded Hanf-sentence of locality radius $r$, we even have 
$\updatetime=k{\cdot}\big(d^{r+1}\big)^{\bigOh(\size{\schema}+d^{r+1})}
\leq 2^{\bigOh(\size{\schema}d^{2r+2})}$.
This completes the proof of Theorem~\ref{thm:AnsweringBooleanQueries}.
\end{proof}

In \cite{FrickGrohe_APAL2004}, Frick and Grohe  obtained a
matching lower bound for answering Boolean $\FO$-queries of schema
$\schema=\set{E}$ 
on databases of degree at most $d\deff 3$ in the static setting. They used the
(reasonable) complexity theoretic assumption $\FPT\neq\AWstar$ and
showed that if this assumption is correct, then there is no algorithm
that answers Boolean $\FO$-queries $\phi$ on $\schema$-dbs $\DB$ of degree
$\leq 3$ in time
$2^{2^{2^{\littleoh(\size{\phi})}}}{\cdot}\poly(\size{\DB})$ in the
static setting
(see Theorem~2 in \cite{FrickGrohe_APAL2004}).
As a consequence, the same lower bound holds in the dynamic setting
and shows that in Theorem~\ref{thm:AnsweringBooleanQueries}, the
3-fold exponential dependency on the query size 
$\size{\phi}$ cannot be substantially lowered (unless $\FPT=\AWstar$):

\begin{corollary}
Let $\schema\deff\set{E}$ and let $d\deff 3$.
If $\FPT\neq\AWstar$, then 
there is no dynamic algorithm that receives a Boolean
$\FO[\schema]$-query $\phi$ and a 
$\schema$-db $\DBstart$, and computes within 
$\preprocessingtime\leq f(\phi){\cdot}\poly(\size{\DBstart})$
preprocessing time a data structure that can be updated in time
$\updatetime\leq f(\phi)$ and allows to
return the query result $\phi(\DB)$ with answer time
$\answertime\leq f(\phi)$, for a function $f$ with 
$f(\phi)=2^{2^{2^{\littleoh(\size{\phi})}}}$.
\end{corollary}

\makeatletter{}%

\section{Technical Lemmas on Types and Spheres Useful for Handling Non-Boolean Queries}
\label{section:TypesAndSpheres}

For our algorithms for evaluating non-Boolean queries it will be
convenient to work with a fixed list of 
representatives of $d$-bounded $r$-types, provided by the following
straightforward lemma. 

\begin{lemma}\label{lemma:isotypes}
There is an algorithm which upon input of a schema $\schema$, a degree
bound $d\geq 2$, a radius $r\geq 0$, and a number $k\geq 1$,
computes a list $\Typeslistrdk= \tau_1,\ldots,\tau_\ell$ (for a
suitable $\ell\geq 1$) of $d$-bounded
$r$-types with $k$ centres (over $\schema$), such that 
for every $d$-bounded $r$-type $\tau$ with $k$ centres (over
$\schema$) there is exactly one $i\in[\ell]$ such that
$\tau\isom\tau_i$.
The algorithm's runtime is \
$2^{(kd^{r+1})^{\bigOh(\size{\schema})}}$.
Furthermore, upon input of a $d$-bounded $r$-type $\tau$ with $k$
centres (over $\schema$), the particular $i\in[\ell]$ with
$\tau\isom\tau_i$ can be computed in time
$2^{(kd^{r+1})^{\bigOh(\size{\schema})}}$. 
\end{lemma}

Throughout the remainder of this paper, $\Typeslistrdk$ will always
denote the list provided by Lemma~\ref{lemma:isotypes}.
The following lemma will be useful for evaluating Boolean combinations
of sphere-formulas.

\begin{lemma}\label{lem:normalform}
  Let $\schema$ be a schema, let $r\geq 0$, $k\geq 1$, $d\geq
  2$, and let $\Typeslistrdk=\tau_1,\ldots,\tau_\ell$.  
\\
  Let $\ov{x}=(x_1,\ldots,x_k)$ be a list of $k$ pairwise distinct variables.
  For every Boolean combination $\psi(\x)$ of $d$-bounded
  sphere-formulas of radius at most $r$ (over $\schema$), there is an
  ${I}\subseteq [\ell]$ such that
  \ \(
     \psi(\ov{x})
     \; \equivd \;
     \Oder_{i \in {I}} \sphere{\tau_i}{\x}.
  \)
\\
  Furthermore, given $\psi(\x)$, the set $I$ can be computed in time 
  $\poly(\size{\psi})\cdot 2^{(kd^{r+1})^{\bigOh(\size{\schema})}}$.
\end{lemma}

\begin{proof}
As a first step, we consider each sphere-formula $\zeta$ that occurs in $\psi$
and replace it by a $d$-equivalent disjunction of sphere-formulas
$\sphere{\tau_j}{\x}$ with $\tau_j$ in $\Typeslistrdk$:
\ if $\zeta$ has arity $k'\leq k$ and radius $r'\leq r$ and is of the
form $\sphere{\rho}{\ov{x}'}$ with
$\ov{x}'=x_{\nu_1},\ldots,x_{\nu_{k'}}$ for $1\leq
\nu_1<\cdots<\nu_{k'}\leq k$ and $\rho=(S,\ov{s})$ with
$\ov{s}=s_1,\ldots,s_{k'}$, then we replace $\zeta$ by the formula
\ $\zeta'\deff
 \Oder_{j\in J}\sphere{\tau_j}{\ov{x}},
$ \
where $J$ consists of all those $j\in[\ell]$ where
for $(T,\ov{t})=\tau_j$ with $\ov{t}=t_1,\ldots,t_k$ and for 
$\ov{t}'\deff t_{\nu_1},\ldots,t_{\nu_{k'}}$ we have
\ $
   \big( S,\, \ov{s}\big)
   \ \isom \ 
   \big( \, \inducedSubStr{T}{\neighb{r'}{T}{\ov{t}'}}, \, \ov{t}'\,\big)
$.
It is straightforward to see that $\zeta'$ and $\zeta$ are
$d$-equivalent. 

Let $\psi_1$ be the formula obtained from $\psi$ by replacing each
$\zeta$ by $\zeta'$. By the Lemmas~\ref{lemma:isotypes} and
\ref{lem:basic_facts}, $\psi_1$ can be constructed in time 
$\bigOh(\size{\psi}\cdot 2^{(kd^{r+1})^{\bigOh(\size{\schema})}})$.
Note that $\psi_1$ is a Boolean combination of formulas 
$\sphere{\tau_j}{\ov{x}}$ for $j\in [\ell]$.

In the second step, we repeatedly use de Morgan's law to push all
$\nicht$-symbols in $\psi_1$ directly in front of sphere-formulas.
Afterwards, we replace 
every subformula of the form $\nicht \sphere{\tau_j}{\x}$ by the
$d$-equivalent formula
\ $
  \Oder_{i\in [\ell]\setminus\set{j}} \sphere{\tau_i}{\x}
 $.  
Let $\psi_2$ be the formula obtained from $\psi_1$ by these transformations.
Constructing $\psi_2$ from $\psi_1$ takes time at most 
$\bigOh({\size{\psi_1}})\cdot 2^{(kd^{r+1})^{\bigOh(\size{\schema})}} = 
\bigOh(\size{\psi}\cdot 2^{(kd^{r+1})^{\bigOh(\size{\schema})}})$.
 
In the third step, we eliminate all the $\und$-symbols in $\psi_2$.
By the definition of the sphere-formulas $\tau_1,\ldots,\tau_\ell$ we have
\begin{equation}
  \label{eq:1}
  \sphere{\tau_i}{\x} \land \sphere{\tau_{i'}}{\x} \quad \equivd\quad
  \begin{cases} 
    \ \sphere{\tau_i}{\x}\text{,} & \text{if \ } i=i' \\
    \ \bot \text{,} & \text{if \ } i\neq i'
  \end{cases}
\end{equation}
where $\bot$ is an unsatisfiable formula.
Thus, by the distributive law we obtain for all $m\geq 1$ and all
$I_1,\ldots,I_m\subseteq[\ell]$ that
\[
  \Und_{j\in [m]}\Big(\Oder_{i\in I_j} \sphere{\tau_i}{\x}\Big)
  \quad\equivd\quad
  \Oder_{i_1\in I_1}\!\!\cdots\!\!\Oder_{i_m\in I_m} \Big( 
    \sphere{\tau_{i_1}}{\x}\und\cdots\und\,\sphere{\tau_{i_m}}{\x}
  \Big)
  \quad\equivd\quad
  \Oder_{i\in I}\sphere{\tau_i}{\x}
\]
for $I\deff I_1\cap\cdots\cap I_m$.
We repeatedly use this equivalence 
during a bottom-up traversal of the syntax-tree of $\psi_2$ to eliminate all
the $\und$-symbols in $\psi_2$. The resulting formula $\psi_3$ is 
obtained in time polynomial in the size of $\psi_2$.
Furthermore, $\psi_3$ is of the desired form
$\Oder_{i\in I}\sphere{\tau_i}{\x}$ for an $I\subseteq[\ell]$. 
The overall time for constructing $\psi_3$ and $I$ is
  $\poly(\size{\psi})\cdot 2^{(kd^{r+1})^{\bigOh(\size{\schema})}}$.
This completes the proof of Lemma~\ref{lem:normalform}.
\end{proof}

For evaluating a Boolean combination $\psi(\ov{x})$ of sphere-formulas \emph{and
Hanf-sentences} on a given $\schema$-db $\DB$, an obvious approach is
to first consider every Hanf-sentence $\chi$ that occurs in $\psi$,
to check if $\DB\models\chi$, and replace every occurrence of
$\chi$ in $\psi$ with $\true$ (resp., $\false$) if $\DB\models\chi$
(resp., $\DB\notmodels\chi$). The resulting formula $\psi'(\x)$ is
then transformed into a disjunction $\psi''(\x)\deff\Oder_{i\in
  I}\sphere{\tau_i}{\x}$ by Lemma~\ref{lem:normalform}, and the query
result $\psi(\DB)=\psi''(\DB)$ is obtained as the union of the query results 
$\sphere{\tau_i}{\DB}$ for all $i\in I$.

While this works well in the static setting (i.e., without database
updates), in the dynamic setting we have to take care of the fact that
database updates might change the status of a Hanf-sentence $\chi$ in
$\psi$,
i.e., an update operation might turn a database $\DB$ with
$\DB\models\chi$ into a database $\DBStrich$ with
$\DBStrich\notmodels\chi$ (and vice versa).
Consequently, the formula $\psi''(\x)$ that is equivalent to
$\psi(\x)$ on $\DB$ might be inequivalent to $\psi(\x)$ on $\DBStrich$.

To handle the dynamic setting correctly, at the end of each update
step we will use the following lemma 
(the lemma's proof is an easy consequence of
Lemma~\ref{lem:normalform}). 

\begin{lemma}\label{lemma:dynamicHanf}
  Let $\schema$ be a schema.
  Let $s\geq 0$ and let $\chi_1$, \ldots, $\chi_s$ be
  $\FOmod[\schema]$-sentences. 
  Let $r\geq 0$, $k\geq 1$, $d\geq
  2$, and let $\Typeslistrdk=\tau_1,\ldots,\tau_\ell$.  
  Let $\ov{x}=(x_1,\ldots,x_k)$ be a list of $k$ pairwise distinct
  variables.
  For every Boolean combination $\psi(\x)$ of the sentences
  $\chi_1,\ldots,\chi_s$ and of $d$-bounded
  sphere-formulas of radius at most $r$ (over $\schema$), and for
  every $J\subseteq [s]$ there is a set
  $I\subseteq [\ell]$ such that
  \[
     \psi_J(\ov{x})
     \quad \equivd \quad
     \Oder_{i \in {I}} \sphere{\tau_i}{\x},
  \]
  where $\psi_J$ is the formula obtained from $\psi$ by replacing
  every occurrence of a sentence $\chi_j$ with $\true$ if $j\in J$ and
  with $\false$ if $j\not\in J$ (for every $j\in[s]$).
\\
  Given $\psi$ and $J$, the set $I$ can be computed in time 
  $\poly(\size{\psi})\cdot 2^{(kd^{r+1})^{\bigOh(\size{\schema})}}$.
\end{lemma}

To evaluate a single sphere-formula $\sphere{\tau}{\x}$ for a given
$r$-type $\tau$ with $k$ centres (over $\schema$),
it will be  useful to decompose $\tau$ into its
connected components as follows.
Let $\tau=(T,\ov{t})$ with $\ov{t}=(t_1,\ldots,t_k)$. 
Consider the Gaifman graph $\GGT$ of $T$ and
let $C_1,\ldots,C_c$ be the vertex sets of the $c$ connected
components of $\GGT$.
For each connected component $C_j$ of $\GGT$,
let $\ov{t}_j$ be the subsequence of $\ov{t}$ consisting of all
elements of $\ov{t}$ that belong to $C_j$, and let $k_j$ be the length
of $\ov{t}_j$.
Since $(T,\ov{t})$ is an $r$-type with $k$ centres, we have
$T=\NrT{\ov{t}}$, and thus $c\leq k$ and $k_j\geq 1$ for all $j\in[c]$.
To avoid ambiguity, we make sure that the list $C_1,\ldots,C_c$ is sorted in such a way that for all $j<j'$ we have $i<i'$ for the smallest $i$ with
 $t_i\in C_j$ and the smallest $i'$ with $t_{i'}\in C_{j'}$.

For each $C_j$ consider the $r$-type with $k_j$ centres $\rho_j=\big(T[C_j],\,\ov{t}_j\big)$.
Let $\nu_j$ be the unique integer
such that $\rho_j$ is isomorphic to the $\nu_j$-th element in
the list $\Typeslistrd{k_j}$, and let $\tau_{j,\nu_j}$ be the
$\nu_j$-th element in this list.

It is straightforward to see that the formula $\sphere{\tau}{\x}$ is
$d$-equivalent to the formula
\begin{equation}\label{eq:decomposition}    
 \connsphere{\tau}{\x}
  \quad \deff \quad
  \Und_{j\in [c]} \sphere{\tau_{j,\nu_j}}{\x_j}
  \ \und \ 
  \Und_{j\neq j'} \nicht \, \dist^{k_j,k_{j'}}_{\leq 2r+1}(\x_j,\x_{j'}),
\end{equation}
where $\x_j$ is the subsequence of $\x$ obtained from $\x$ in the same
way as $\ov{t}_j$ is obtained from $\ov{t}$, and
$\dist^{k_j,k_{j'}}_{\leq 2r+1}(\x_j,\x_{j'})$ is a 
formula of schema $\schema$ which expresses that for some variable $y$
in $\x_j$ and some variable $y'$ in $\x_{j'}$ the distance between 
$y$ and $y'$ is $\leq 2r{+}1$. I.e., for $\ov{a}=(a_1,\ldots,a_{k_j})$
and $\ov{b}=(b_1,\ldots,b_{k_{j'}})$ we have
$(\ov{a},\ov{b})\in\dist^{k_j,k_{j'}}_{\leq 2r+1}(\DB) \iff
\dist^{\DB}(\ov{a};\ov{b})\leq 2r{+}1$, where 
\begin{equation}\label{eq:dist-of-tuples}
\text{
  $\dist^{\DB}(\ov{a};\ov{b})\leq 2r{+}1$ means that
  $\dist^{\DB}(a_i,b_{i'})\leq 2r{+}1$ for some $i\in[k_j]$ and
  $i'\in[k_{j'}]$.
}
\end{equation}
Using the Lemmas~\ref{lem:basic_facts} and \ref{lemma:isotypes}, the
following lemma is straightforward.

\begin{lemma}\label{lem:decomposition}
There is an algorithm which upon input of a
schema $\schema$, numbers $r\geq 0$, $k\geq 1$, and $d\geq 2$, and
an $r$-type $\tau$ with $k$ centres (over $\schema$) computes the formula $\connsphere{\tau}{\x}$, along with the corresponding 
parameters $c$ and 
$k_{j}$, $\nu_{j}$, $\ov{x}_j$, $\tau_{j,\nu_j}$ for all
$j\in[c]$. 
\\
The algorithm's runtime is $2^{(kd^{r+1})^{\bigOh(\size{\schema})}}$.
\end{lemma}

We define the \emph{signature of $\tau$} to be the tuple
$\sgn{\tau}$ built from the parameters $c$ and
$\big(k_j,\nu_j,\setc{\mu\in[k]}{x_\mu\text{ belongs to
  }\ov{x}_j}\big)_{j\in[c]}$ 
obtained from the above lemma.
The signature $\sgnDB{\ov{a}}$ of a tuple $\ov{a}$ in a database $\DB$
(w.r.t.\ radius $r$)
is defined as $\sgn{\rho}$ for $\rho\deff\big(\NrD{\a},\a\big)$.
Note that $\a\in\sphere{\tau}{\DB} \iff \sgnDB{\a}=\sgn{\tau}$.

\makeatletter{}%

\section{Testing Non-Boolean \fomodtext Queries Under Updates}
\label{section:testing}

This section is devoted to the proof of the following theorem.

\begin{theorem}\label{thm:testing}
There is a dynamic algorithm that receives a schema $\schema$, a
degree bound $d\geq 2$, a $k$-ary $\FOmod[\schema]$-query $\phi(\ov{x})$ (for
some $k\in\NN$), and a
$\schema$-db $\DBstart$ of degree $\leq d$, and computes within 
$\preprocessingtime= f(\phi,d)\cdot\size{\DBstart}$
preprocessing time a data structure that can be updated in time
$\updatetime= f(\phi,d)$ and allows to
test for any input tuple 
$\ov{a}\in\Dom^k$ whether $\ov{a}\in\phi(\DB)$ within testing time
$\testingtime= \bigOh(k^2)$.
The function $f(\phi,d)$ is of the form 
$2^{d^{2^{\bigOh(\size{\phi})}}}$.
\end{theorem}

\newcommand{\testingupdatetime}{\text{\texttt{\upshape\small testingupdatetime}}}
\newcommand{\booleanupdatetime}{\text{\texttt{\upshape\small
      booleanupdatetime}}}
\newcommand{\hanftime}{\text{\texttt{\upshape\small hanftime}}}

\newcommand{\hanfform}{\psi}
\newcommand{\hanfsentenceset}{\Psi_{\text{\upshape Bool}}}
\newcommand{\hanfsentence}{\psi_{\text{\upshape Bool}}}
\newcommand{\hanfsubset}{\Phi}
\newcommand{\connectedtuples}{C}
\newcommand{\relF}{F}

For the proof, we use the lemmas provided in
Section~\ref{section:TypesAndSpheres} and the following lemma.

\begin{lemma}\label{lem:testing}
There is a dynamic algorithm that receives a schema $\schema$, a
degree bound $d\geq 2$, numbers $r\geq 0$ and $k\geq 1$, 
an $r$-type $\tau$ with $k$ centres (over $\schema$),
and a
$\schema$-db $\DBstart$ of degree $\leq d$, and computes within 
$\preprocessingtime=  2^{(kd^{r+1})^{\bigOh(\size{\schema})}} \cdot\size{\DBstart}$
preprocessing time a data structure that can be updated in time
$\updatetime=  2^{(kd^{r+1})^{\bigOh(\size{\schema})}}$ and allows to
test for any input tuple 
$\ov{a}\in\Dom^k$ whether $\ov{a}\in\sphere{\tau}{\DB}$ within testing time
$\testingtime= \bigOh(k^2)$.
\end{lemma}
\begin{proof}
The preprocessing routine starts by using Lemma~\ref{lem:decomposition} to
compute the formula $\connsphere{\tau}{\x}$, along with the according 
parameters $c$ and $k_{j}$, $\nu_{j}$, $\ov{x}_j$, $\tau_{j,\nu_j}$
for each $j\in[c]$. 
This is done in time $2^{(kd^{r+1})^{\bigOh(\size{\schema})}}$.
We let $\sgn{\tau}$ be the signature of $\tau$ (defined directly after Lemma~\ref{lem:decomposition}).
Recall that \ $\connsphere{\tau}{\x}\ \equivd\ \sphere{\tau}{\x}$, \
and recall from equation~\eqref{eq:decomposition} the precise definition of the
formula $\connsphere{\tau}{\x}$.
Our data structure will store the following information on the
database $\DB$:
\begin{itemize}
 \item the set $\Gamma$ of all tuples $\ov{b}\in\adom{\DB}^{k'}$ where
   $k'\leq k$ and $\NrD{\ov{b}}$ \emph{is connected}, and
 \item 
   for every $j\in[c]$ and every tuple $\ov{b}\in \Gamma$ of arity
   $k_j$, the unique number $\nu_{\ov{b}}$ such that 
   $\rho_{\ov{b}}\deff \big(\NrD{\ov{b}},\ov{b}\big)$ is isomorphic to
   the $\nu_{\ov{b}}$-th element in the list
   $\Typeslistrd{k_j}$.
\end{itemize}

\noindent
We want to store this information in such a way that for any
given tuple $\ov{b}\in\Dom^{k'}$ it can be
checked in time $\bigOh(k)$ 
whether $\ov{b}\in\Gamma$.
To ensure this, we use a $k'$-ary array $\mathtt{\Gamma}_{k'}$ that is initialised
to 0, and where during update operations the entry
$\mathtt{\Gamma}_{k'}[\ov{b}]$ is set to 1 for all $\ov{b}\in
\Gamma$ of arity $k'$. 
In a similar way we can ensure that for any given $j\in[c]$ and any
$\ov{b}\in\Gamma$ of arity $k_j$, the number $\nu_{\ov{b}}$ can be
looked up in time $\bigOh(k)$.

The $\TEST$ routine upon input of a tuple $\ov{a}=(a_1,\ldots,a_k)$
proceeds as follows.

First, we partition $\ov{a}$ into $\ov{a}_1,\ldots,\ov{a}_{c'}$ (for
$c'\leq k$) such that 
$C_j\deff\nrD{\ov{a}_j}$ for $j\in[c']$ are the connected components
of $\NrD{\ov{a}}$. As in the definition of the formula
$\connsphere{\tau}{\x}$, we make sure that this list is sorted in such
a way that for all $j<j'$ we have $i<i'$ for the smallest $i$ with
$a_i\in C_j$ and the smallest $i'$ with $a_{i'}\in C_{j'}$.
All of this can be done in time $\bigOh(k^2)$ by first constructing the graph $H$
with vertex set $[k]$ and where there is an edge between vertices $i$ and $j$ iff
the tuple $(a_i,a_j)$ belongs to $\Gamma$, and then computing the
connected components of $H$.

Afterwards, for each $j\in[c']$ we use time $\bigOh(k)$ to look up the
number $\nu_{\ov{a}_j}$.
We then let $\sgnDB{\ov{a}}$ be the tuple built from $c'$ and 
$\big(|\ov{a}_j|,\nu_{\ov{a}_j},\setc{\mu\in[k]}{a_\mu\text{ belongs
    to }\ov{a}_j}\big)_{j\in[c']}$.
It is straightforward to see that $\ov{a}\in\connsphere{\tau}{\DB}$
iff $\sgnDB{\ov{a}}=\sgn{\tau}$. Therefore, the $\TEST$ routine checks
whether  $\sgnDB{\ov{a}}=\sgn{\tau}$ and outputs ``yes'' if this is
the case and ``no'' otherwise.
The entire time used by the $\TEST$ routine is
$\testingtime=\bigOh(k^2)$.

To finish the proof of Lemma~\ref{lem:testing}, we have to give
further details on the $\PREPROCESS$ routine and the $\UPDATE$
routine.
The $\PREPROCESS$ routine initialises $\Gamma$ as the empty set
$\emptyset$ and then performs $\card{\DBstart}$ update operations to
insert all the tuples of $\DBstart$ into the data structure.
The $\UPDATE$ routine proceeds as follows.

Let $\DBold$ be the database before the update is received and
let $\DBnew$ be the database after the update has been performed.
Let the update command be of the form $\Update\,R(a_1,\ldots,a_{\ar(R)})$.
We let $r'\deff r+(\ar(R){-}1)(2r{+}1)$.
All elements whose $r'$-neighbour\-hood might have changed belong to
the set 
\ $\UpdateSet\deff\neighb{r'}{\DBStrich}{\ov a}$, \ where
$\DBStrich\deff\DBnew$ if the update command is $\Insert\,R(\ov{a})$,
and $\DBStrich\deff\DBold$ if the update command is $\Delete\,R(\ov{a})$.

According to
Lemma~\ref{lem:basic_facts}(\ref{eq:Nconnk:lem:basic_facts}), all
tuples $\ov{b}$ that have to be inserted into or deleted from $\Gamma$
are built from elements in $\UpdateSet$. To update the information
stored in our data structure, we loop
through all tuples of arity $\leq k$ that are built from elements in $\UpdateSet$.

Using Lemma~\ref{lem:basic_facts}, we obtain that $|U|\leq
\ar(R){\cdot}d^{r'+1}$. The number of candidate tuples $\ov{b}$ built
from elements in $\UpdateSet$ is at most $ \big(\ar(R){\cdot}d^{r'+1}\big)^{k+1}$.
Using the Lemmas~\ref{lem:basic_facts} and \ref{lemma:isotypes},
it is not difficult to see that the entire update time is at most $\updatetime
= 2^{(kd^{r+1})^{\bigOh(\size{\schema})}}$.
The initialisation time $\inittime$ is of the same form, and hence the
preprocessing time is as claimed in the lemma.
This completes the proof of Lemma~\ref{lem:testing}.
\end{proof}

\medskip

Theorem~\ref{thm:testing} is now obtained by combining
Theorem~\ref{thm:HNF}, Lemma~\ref{lem:testing}, Theorem~\ref{thm:AnsweringBooleanQueries}, and
Lemma~\ref{lemma:dynamicHanf}.

\begin{proof}[Proof of Theorem~\ref{thm:testing}]
For $k=0$, the theorem immediately follows from
Theorem~\ref{thm:AnsweringBooleanQueries}.
Consider the case where $k\geq 1$. As in the proof of 
Theorem~\ref{thm:AnsweringBooleanQueries}, we assume w.l.o.g.\ that
all the symbols of $\schema$ occur in $\phi$. We start the
preprocessing routine by using Theorem~\ref{thm:HNF} to
transform $\phi(\x)$ into a $d$-equivalent query $\psi(\x)$ in Hanf normal
form; this takes time $2^{d^{2^{\bigOh(\size{\phi})}}}$.
The formula $\psi$ is a Boolean combination of $d$-bounded Hanf-sentences 
and sphere-formulas (over $\schema$) of locality radius
at most $r\deff 4^{\qr(\phi)}$, and each sphere-formula is of arity at most $k$. 
Let $\chi_1,\ldots,\chi_s$ be the list of all Hanf-sentences that
occur in $\psi$.

We use Lemma~\ref{lemma:isotypes} to compute the list $\Typeslistrdk=
\tau_1,\ldots,\tau_\ell$.
In parallel for each $i\in[\ell]$, we use the 
algorithm
provided by Lemma~\ref{lem:testing} for $\tau\deff \tau_i$.
Furthermore, for each $j\in[s]$, we use the 
algorithm provided
by Theorem~\ref{thm:AnsweringBooleanQueries} upon input of the
Hanf-sentence $\phi\deff\chi_j$.
In addition to the components used by these dynamic algorithms, our data
structure also stores 
\begin{itemize}
 \item the set $J\deff\setc{j\in[s]}{\DB\models\chi_j}$, 
 \item the particular set $I\subseteq [\ell]$ provided by
   Lemma~\ref{lemma:dynamicHanf} for $\psi(\ov{x})$ and $J$, and
 \item the set $K=\setc{\sgn{\tau_i}}{i\in I}$, where for each
   type $\tau$, \,$\sgn{\tau}$ is the signature of $\tau$ defined
   directly after Lemma~\ref{lem:decomposition}.
\end{itemize}
 
The $\TEST$ routine upon input of a tuple $\ov{a}=(a_1,\ldots,a_k)$
proceeds in the same way as in the proof of Lemma~\ref{lem:testing} to
compute 
in time $\bigOh(k^2)$ the signature $\sgnDB{\ov{a}}$ of the tuple
$\ov{a}$.
For every $i\in[\ell]$ we have 
$\ov{a}\in\sphere{\tau_i}{\DB}\iff\sgnDB{\ov{a}}=\sgn{\tau_i}$.
Thus, $\ov{a}\in\phi(\DB) \iff \sgnDB{\ov{a}}\in K$.
Therefore, the $\TEST$ routine checks
whether  $\sgnDB{\ov{a}}\in K$ and outputs ``yes'' if this is
the case and ``no'' otherwise.
To ensure that this test can be done in
time $\bigOh(k^2)$, we use an array construction for storing $K$
(similar to the one for storing $\Gamma$ in the proof of
Lemma~\ref{lem:testing}).

The $\UPDATE$ routine runs in parallel the update routines for all the
used dynamic data structures. Afterwards, it recomputes $J$ by calling
the $\ANSWER$ routine for $\chi_j$ for all $j\in[s]$. Then, it uses 
Lemma~\ref{lemma:dynamicHanf} to recompute $I$.
The set $K$ is then recomputed by applying Lemma~\ref{lem:decomposition}
for $\tau\deff\tau_i$ for all $i\in I$.
It is straightforward to see that the overall runtime of the $\UPDATE$
routine is $\updatetime = 2^{d^{2^{\bigOh(\size{\schema})}}}$.
This completes the proof of Theorem~\ref{thm:testing}.
\end{proof}

\makeatletter{}%

\section{Representing Databases by Coloured Graphs}
\label{section:dbtograph}

\newcommand{\upsteps}{\ensuremath{d^{\bigoh(k^2r + k\size{\schema})}}\xspace}
\newcommand{\upstepstime}{\ensuremath{2^{\bigOh(\size{\sigma}k^2d^{2r+2})}}\xspace}
\newcommand{\aritys}{s}

To obtain dynamic algorithms for counting and enumerating query
results, it will be convenient to work with a representation of
databases by coloured graphs that is similar to the representation 
used in \cite{DBLP:conf/pods/DurandSS14}.
For defining this representation, let us consider a fixed $d$-bounded $r$-type
$\tau$ with $k$ centres (over a schema $\schema$).
Use Lemma~\ref{lem:decomposition} to compute the formula
$\connsphere{\tau}{\ov{x}}$ (for $\ov{x}=(x_1,\ldots,x_k)$) 
and the according parameters $c$ and
$k_j,\nu_j,\ov{x}_j,\tau_{j,\nu_j}$, and let $\sgn{\tau}$ be the
signature of $\tau$.
To keep the notation simple, we assume w.l.o.g.\ that
$\ov{x}_1=x_1,\ldots,x_{k_1}$,
$\ov{x}_2=x_{k_1+1},\ldots,x_{k_1+k_2}$ etc. 

Recall that $\sphere{\tau}{\x}$ is $d$-equivalent to the formula 
\[
 \connsphere{\tau}{\x}
  \quad \deff \quad
  \Und_{j\in [c]} \sphere{\tau_{j,\nu_j}}{\x_j}
  \ \und \ 
  \Und_{j\neq j'} \nicht \, \dist^{k_j,k_{j'}}_{\leq
    2r+1}(\x_j,\x_{j'}).
\]
To count or enumerate the results of the formula $\sphere{\tau}{\x}$
we represent the database $\DB$ by a \emph{$c$-coloured graph}
$\GofD$.
Here, a \emph{$c$-coloured graph} $\graph$ is a database of the particular schema 
\[
  \schema_c\quad \deff\quad \set{\relE,\relT_1,\ldots,\relT_c},
\]
where $\relE$ is a binary
relation symbol and $\relT_1,\ldots,\relT_c$ are unary relation
symbols.
We define $\GofD$ in such a way that
the task of counting or enumerating the results of the query
$\sphere{\tau}{\ov{x}}$ on the database $\DB$ can be reduced to
counting or enumerating the results of the query 
  \begin{equation}
    \label{eq:phi-ell}
    \phi_c(z_1,\ldots,z_c) \quad\deff\quad
    \Und_{j\in[c]} \relT_j(z_j)
    \ \land \  
    \Und_{j\neq j'} \neg \,\relE(z_j,z_{j'})
  \end{equation}
on the $c$-coloured graph $\GofD$.
The vertices of $\GofD$ correspond to tuples
over $\adom{\DB}$ whose $r$-neighbourhood is connected;
a vertex has colour $\relT_j$ if its associated tuple $\ov{a}$ is in 
$\sphere{\tau_{j,\nu_j}}{\DB}$; and an edge
between two vertices indicates that 
$\dist^{\DB}(\ov{a};\ov{b})\leq 2r{+}1$, for
their associated tuples $\ov{a}$ and $\ov{b}$.
The following lemma allows to translate a dynamic algorithm for counting or
enumerating the results of the query $\phi_c(z_1,\ldots,z_c)$ on
$c$-coloured graphs
into a dynamic algorithm for counting or enumerating
the result of the query $\sphere{\tau}{\x}$ on $\DB$.

\begin{lemma}\label{lem:translationtograph}
Suppose that the counting problem
(the enumeration problem) 
for $\phi_c(\z)$ on $\schema_c$-dbs of degree
at most $d'$ can be solved by a dynamic algorithm 
with initialisation time $\inittime(c,d')$, update time 
$\updatetime(c,d')$, and counting time $\countingtime(c,d')$ (delay
$\delaytime(c,d')$).
Then for every schema $\schema$ and every $d\geq 2$ the following holds.
\begin{enumerate}[(1)]
\item\label{eq:lem:db2graph:part1} 
Let $r\geq 0$, $k\geq 1$, $\tau$ a
$d$-bounded $r$-type  with $k$ centres, and
fix $d'\deff d^{2k^2(2r+1)}$ and
$\widetilde{t_x}\deff\max^k_{c=1}t_x(c,d')$ for
$t_x\in\{\inittime,\updatetime,\countingtime,\delaytime\}$.
The counting problem (the enumeration problem) for
$\sphere{\tau}{\x}$ on $\schema$-dbs of degree 
at most $d$ can be solved by a dynamic algorithm 
with counting time $\widetilde{\countingtime}$ (delay $\bigOh(\widetilde{\delaytime}
k)$),  update time 
 $\updatetimeStrich\leq \widetilde{\updatetime}\upsteps + \upstepstime$, and 
 initialisation time $\widetilde{\inittime}$.

\item\label{eq:lem:db2graph:part2}
The counting problem (the enumeration problem) for $k$-ary $\FOmod$-queries
$\phi(\x)$ on $\schema$-dbs of degree 
at most $d$ can be solved 
with counting time $\bigOh(1)$ (delay $\bigOh(\widehat{\delaytime}
k)$),  update time 
 $(\widehat{\updatetime}+\widehat{\countingtime})2^{d^{2^{\bigOh(\size{\phi})}}}$, and 
 initialisation time
 $\widehat{\inittime}2^{d^{2^{\bigOh(\size{\phi})}}}$
 where
 $\widehat{t_x}=\max^k_{c=1}t_x\bigl(c,d^{2^{\bigOh(\size{\phi})}}\bigr)$ for $t_x\in\{\inittime,\updatetime,\countingtime,\delaytime\}$.
\end{enumerate}
\end{lemma}

 \begin{proof}
We prove 
part~\eqref{eq:lem:db2graph:part1}
by a reduction from $\connsphere{\tau}{\x}$
to $\phi_c$. We use the notation introduced at the beginning of
Section~\ref{section:dbtograph},
and we let $\tau_j\deff\tau_{j,\nu_j}$ for every $j\in[c]$.
For a $\schema$-db $\DB$ we let $\GofD$ be the $\schema_c$-db with
  \begin{align*}
   \relT_j^{\GofD} &\ \ := \ \ \setc{\,v_{\a}}{\ov{a}\in \adom{\DB}^{k_j}\text{
     with }\big(\nb{r}{\DB}(\a),\a\big) \isomorph \tau_j\, }, \quad \text{for all
   $j\in[c]$, \ \ and}  \\
   \relE^{\GofD} &\ \ := \ \  \setc{\, (v_{\a},v_{\b})\in V^2}{
     \dist^{\DB}(\ov{a};\ov{b})\leq 2r{+}1 \,},
  \end{align*}
where \ $V \deff \bigcup_{j\in[c]}\relT_j^{\GofD}$.
We will shortly write $\relE$ and $\relT_j$ instead
of $\relE^{\GofD}$ and $\relT_j^{\GofD}$.

Using Lemma~\ref{lem:basic_facts} (and the fact that $\tau_j$ is
connected) we obtain that $(v_{\a},v_{\b})\in E$ iff
$\nb{r}{\DB}(\a,\b)$ is connected. If 
$\nb{r}{\DB}(\a,\b)$ is connected, then 
$\b  \in \bigl(N_{r+(|\a|+|\b|-1)(2r+1)}^{\DB}(a_1)\bigr)^{|\b|}$.
It follows that the degree of $\GofD$ is bounded by $d^{2k^2(2r+1)}$. 
Furthermore, by the definition of $\GofD$ and $\phi_c$ we get that
\ $(\a_1,\ldots,\a_c) \in \sphere{\tau}{\DB}$ $\iff$
$(v_{\a_1},\ldots,v_{\a_c}) \in \phi_c(\GofD)$, \ 
for all tuples $\ov{a}_1,\ldots,\ov{a}_c$ where $\ov{a}_j$ has
arity $k_j$ for each $j\in[c]$.
As a consequence,
\ $
  \setsize{\sphere{\tau}{\DB}} \, = \, \setsize{\phi_c(\GofD)},
$ \ 
and we can
therefore use the $\COUNT$ routine for $\phi_c$ on $\GofD$ to count the
number of tuples in $\sphere{\tau}{\DB}$. 
Furthermore, for each tuple
$(v_{\a_1},\ldots,v_{\a_c})\in \phi_c(\GofD)$ we can compute
$(\a_1,\ldots,\a_c)$ in time $\bigoh(k)$.  Therefore, given an
$\ENUMERATE$ routine for $\phi_c(\GofD)$ with delay $\delaytime$ we can
produce an enumeration of $\sphere{\tau}{\DB}$ with delay
$\bigoh(\delaytime k)$.

It remains to show how to construct 
and maintain $\GofD$ when the database $\DB$ is updated.
As initialisation for the empty database $\DBempty$ we just perform
the $\INIT$ routine of the dynamic algorithm for
$\phi_c(\z)$ on $\schema_c$-dbs of degree at most $d'$.
The $\UPDATE$ routine of the dynamic algorithm for $\sphere{\tau}{\x}$ on
$\schema$-dbs of degree at most $d$ is provided by the following claim.

\begin{claim}\label{claim:update}
 If $\DBnew$ is obtained from $\DBold$ by one update step, then
 $\GofDnew$ can be obtained from $\GofDold$ by $\upsteps$ update steps
 and additional computing time $\upstepstime$.
\end{claim}

\begin{proof}
Let the update command be of the form
$\Update\,R(a_1,\ldots,a_{\ar(R)})$ with $\a=(a_1,\ldots,a_{\ar(R)})$.
Let $r'= r+(k{-}1)(2r{+}1)$. 
Let $\DBStrich\in\set{\DBold,\DBnew}$ be the database whose
relation $R$ contains the tuple $\a$ (either
before deletion or after insertion).
Note that all elements whose $r'$-neighbourhood
might have changed, belong to the set
$\UpdateSet\deff\neighb{r'}{\DBStrich}{\ov a}$.

For every $j\in[c]$ and every tuple $\b$ of arity at most $k$ of elements in $\UpdateSet$,
we check whether the $r$-type $\big(\nb{r}{\DBnew}(\b),\b\big)$ of
$\ov{b}$ is isomorphic to $\tau_j$.
Depending on the outcome of this test, we include or exclude $v_{\b}$
from the relation $\relT_j$.
Note that it indeed suffices to consider the tuples $\ov{b}$ built
from elements in $\UpdateSet$: \ The $r$-type of some
tuple $\b$ is changed by the update command only if
$\neighb{r}{\DBStrich}{\ov{b}}$ contains some element from $\a$. 
Furthermore, we only have to consider tuples $\ov{b}$ whose
$r$-neighbourhood
$\Neighb{r}{\DBStrich}{\ov{b}}$ is connected. Using 
Lemma~\ref{lem:basic_facts}(\ref{eq:Nconnk:lem:basic_facts}), we therefore obtain that each component
of $\ov{b}$ belongs to $\neighb{r'}{\DBStrich}{\ov{a}}=\UpdateSet$.

Afterwards, we update the coloured graph's edge relation $\relE$: \
We consider all tuples $\ov{b}$
and $\ov{b}'$ of arity $\leq k$ built from elements in $\UpdateSet$, and
check whether (1) there is a $j\in[c]$ such that $\ov{b}\in \relT_j$, (2)
there is a $j'\in[c]$ such that $\ov{b}'\in \relT_{j'}$,
and (3)
$\dist^{\DBnew}(\ov{b};\ov{b}')\leq 2r{+}1$. If all three checks
return the result ``yes'', then we insert the tuple
$\big(v_{\ov{b}},v_{\ov{b}'}\big)$ into $\relE$, otherwise we remove it
from $\relE$.

It remains to analyse the runtime of the described update procedure.
By Lemma~\ref{lem:basic_facts}, 
 $|\UpdateSet|\leq \ar(R)d^{r'+1}\leq 
\size{\schema}{d^{k(2r+1)}} \leq d^{\bigOh(kr + \lg\size{\schema})}
\leq d^{\bigOh(kr + \size{\schema})}$.
Furthermore, $\UpdateSet$ can be computed
in time $\big(\ar(R)d^{r'+1}\big)^{\bigOh(\size{\schema})}\ \leq \
d^{\bigOh(kr\size{\schema} + \size{\schema}^2)}$.
The number of tuples $\ov{b}$ that we have to consider is at most
$|\UpdateSet|^{k+1}\leq d^{\bigOh(k^2r+k\size{\sigma})}$. 

For each such $\ov{b}$ we use Lemma~\ref{lem:basic_facts}(\ref{eq:Nisom:lem:basic_facts}) to check in
time $2^{\bigOh(\size{\schema}k^2d^{2r+2})}$ \
whether the $r$-type of
$\ov{b}$ is isomorphic to $\tau_j$, for some $j\in[c]$.
In summary, for updating the sets $\relT_1,\ldots,\relT_c$ we use
at most $c|\UpdateSet|^{k+1}\leq d^{\bigOh(k^2r+k\size{\sigma})} $ calls of the $\UPDATE$ routine of
the dynamic algorithm on coloured graphs, and in addition to that we
use computation time at most $2^{\bigOh(\size{\schema}k^2d^{2r+2})}$.

By a similar reasoning 
we obtain that also the edge relation
$\relE$ can be updated by at most 
$d^{\bigOh(k^2r+k\size{\sigma})} $ calls of the $\UPDATE$ routine of
the dynamic algorithm on coloured graphs and additional
computation time at most $2^{\bigOh(\size{\schema}k^2d^{2r+2})}$.
For this note that we can use and maintain an additional array that
allows us to check, for any $a_i$ and $b_j$, in constant time whether
$\dist^{\DB}(a_i,b_j)\leq 2r{+}1$. 
This completes the proof of Claim~\ref{claim:update}.
\end{proof}

Finally, the $\PREPROCESS$ routine of the dynamic algorithm for
$\sphere{\tau}{\x}$ 
proceeds in the obvious way by first calling the $\INIT$ routine for 
$\DBempty$ and then performing $|\DBstart|$ update steps to
insert all the tuples of $\DBstart$ into the data structure.
This completes the proof of part~\eqref{eq:lem:db2graph:part1} of Lemma~\ref{lem:translationtograph}.

\bigskip

We now turn to the proof of part~\eqref{eq:lem:db2graph:part2} of
Lemma~\ref{lem:translationtograph}. 
For $k=0$, the result follows immediately from
Theorem~\ref{thm:AnsweringBooleanQueries}. 
Consider the case where $k\geq 1$. 
W.l.o.g.\ we assume that
all the symbols of $\schema$ occur in $\phi$ (otherwise, we remove from
$\schema$ all symbols that do not occur in $\phi$). We start the
preprocessing routine by using Theorem~\ref{thm:HNF} to
transform $\phi(\x)$ into a $d$-equivalent query $\psi(\x)$ in Hanf normal
form; this takes time $2^{d^{2^{\bigOh(\size{\phi})}}}$.
The formula $\psi$ is a Boolean combination of $d$-bounded Hanf-sentences 
and sphere-formulas (over $\schema$) of locality radius
at most $r\deff 4^{\qr(\phi)}$, and each sphere-formula is of arity at most $k$. 
Note that for $d'\deff d^{2k^2(2r+1)}$ as used in 
the lemma's part~\eqref{eq:lem:db2graph:part1},
it holds that $d'=d^{2^{\bigOh(\size{\phi})}}$.
Let $\chi_1,\ldots,\chi_s$ be the list of all Hanf-sentences that
occur in $\psi$ (recall that $s\leq 2^{d^{2^{\bigOh(\size{\phi})}}}$).

We use Lemma~\ref{lemma:isotypes} to compute the list $\Typeslistrdk=
\tau_1,\ldots,\tau_\ell$ (note  that $\ell \leq 2^{d^{2^{\bigOh(\size{\phi})}}}$).
In parallel for each $i\in[\ell]$, we use the dynamic algorithm
for $\sphere{\tau_i}{\x}$ provided from the lemma's part~(1).
Furthermore, for each $j\in[s]$, we use the dynamic algorithm provided
by Theorem~\ref{thm:AnsweringBooleanQueries} upon input of the
Hanf-sentence $\phi\deff\chi_j$.
In addition to the components used by these dynamic algorithms, our data
structure also stores 
\begin{itemize}
 \item the set $J\deff\setc{j\in[s]}{\DB\models\chi_j}$, 
 \item the particular set $I\subseteq [\ell]$ provided by
   Lemma~\ref{lemma:dynamicHanf} for $\psi(\ov{x})$ and $J$, and
 \item the cardinality $\mycount = |\phi(\DB)|$ of the query result.
\end{itemize}

The $\COUNT$ routine simply outputs the value $\mycount$ in time
$\bigOh(1)$.
The $\ENUMERATE$ routine runs  the $\ENUMERATE$ routine on
$\sphere{\tau_i}{\DB}$ for every $i\in I$.
Note that this enumerates, without repetition, all tuples
in $\phi(\DB)$, because by Lemma~\ref{lemma:dynamicHanf}, 
$\phi(\DB)$ is the union of the sets $\sphere{\tau_i}{\DB}$
for all $i\in I$, and this is a union of pairwise disjoint sets.
The $\UPDATE$ routine runs in parallel the update routines for all
used dynamic data structures. Afterwards, it recomputes $J$ by calling
the $\ANSWER$ routine for $\chi_j$ for all $j\in[s]$. Then, it uses 
Lemma~\ref{lemma:dynamicHanf} to recompute $I$.
The number  $\mycount$ is then recomputed by letting
$\mycount=\sum_{i\in I}\mycount_i$, where 
$\mycount_i$ is the result of the $\COUNT$ routine for $\tau_i$.
It is straightforward%
to verify that the overall runtime of the $\UPDATE$
routine is bounded by
$(\widehat{\updatetime}+\widehat{\countingtime})2^{d^{2^{\bigOh(\size{\phi})}}}$. 
 \end{proof}

\makeatletter{}%

\section{Counting Results of  \fomodtext Queries Under Updates}
\label{section:counting}

This section is devoted to the proof of the following theorem.

\begin{theorem}\label{thm:counting}
There is a dynamic algorithm that receives a schema $\schema$, a
degree bound $d\geq 2$, a $k$-ary $\FOmod[\schema]$-query $\phi(\ov{x})$ (for
some $k\in\NN$), and a
$\schema$-db $\DBstart$ of degree $\leq d$, and computes within 
$\preprocessingtime= f(\phi,d)\cdot\size{\DBstart}$
preprocessing time a data structure that can be updated in time
$\updatetime= f(\phi,d)$ and allows to
return the cardinality $|\phi(\DB)|$ of the query result within
time $\bigOh(1)$.
\\
The function $f(\phi,d)$ is of the form 
$2^{d^{2^{\bigOh(\size{\phi})}}}$.
\end{theorem}

The theorem follows immediately from
Lemma~\ref{lem:translationtograph}\eqref{eq:lem:db2graph:part2} and
the following dynamic counting algorithm for the query $\phi_c(\ov{z})$.

\nc{\updatetimeCountPhi}{\Odc}
\begin{lemma}\label{lem:countPhic}
There is a dynamic algorithm that receives a number $c\geq 1$,
a degree bound $d\geq 2$, and a $\schema_c$-db $\graphStart$ of degree
$\leq d$, and computes $|\phi_c(\graph)|$ with
$\updatetimeCountPhi$ initialisation time, $\bigOh(1)$ counting time,
and $\updatetimeCountPhi$ update time. 
\end{lemma}
\begin{proof}
Recall that 
\ $
    \phi_c(z_1,\ldots,z_c) \ = \
    \Und_{i\in[c]} \relT_i(z_i)
    \ \land \  
    \Und_{j\neq j'} \neg \,\relE(z_{j},z_{j'})
$. \
For all $j,j'\in[c]$ with $j\neq j'$ consider the formula
\ $\theta_{j,j'}(z_1,\ldots,z_c) \ \deff \
 \relE(z_j,z_{j'}) \ \und \ \Und_{i\in[c]}\relT_i(z_i)$.
Furthermore, let \ 
$\alpha(z_1,\ldots,z_c) \ \deff \ \Und_{i\in[c]}\relT_i(z_i)$.
Clearly, for every $\schema_c$-db $\graph$ we have
\begin{eqnarray*}
  \alpha(\graph)
& =
& \relT_1^\graph \times \cdots \times \relT_c^\graph, \qquad
\medskip\\
  \phi_c(\graph) 
& = 
& \alpha(\graph) \ \setminus \ \Big( \bigcup_{j\neq j'} \theta_{j,j'}(\graph)\Big),
\quad\text{and hence,}\quad
  |\phi_c(\graph)| 
 \ \; = \ \;
 |\alpha(\graph)| \  - \ 
  \Big| \bigcup_{j\neq j'} \theta_{j,j'}({\graph})\Big|.
\end{eqnarray*}
By the \emph{inclusion-exclusion principle} we obtain for
$J\deff\setc{(j,j')}{j,j'\in[c],\ j\neq j'}$ that
\[
  \Big| \bigcup_{j\neq j'} \theta_{j,j'}({\graph})\Big|
  \ \ \ = \ \ \
  \sum_{\emptyset\neq K\subseteq J} 
   (-1)^{|K|-1}\ \Big|\!\bigcap_{(j,j')\in K}\theta_{j,j'}({\graph})
   \Big|
 \ \ \ = \ \ \
 \sum_{\emptyset\neq K\subseteq J} 
   (-1)^{|K|-1}\ \big| \phi_K(\graph)\big|
\]
for the formula \ 
$\phi_K(z_1,\ldots,z_c) \ \deff \ \Und_{i\in[c]}\relT_i(z_i) \ \und \ 
  \Und_{(j,j')\in K} E(z_j,z_{j'})$. 
\\
Our data structure stores the following values:
\begin{itemize}
 \item $|\relT_i^\graph|$, \ for each $i\in[c]$, \quad and \quad
  $\mycount_1 \ \deff \ |\alpha(\graph)| \ = \ \prod_{i\in[c]}|\relT_i^{\graph}|$,
 \item $|\phi_K(\graph)|$, \ for each $K\subseteq J$ with
   $K\neq\emptyset$, \quad and
 \item $\mycount_2 \ \deff \
  \sum_{\emptyset\neq K\subseteq J} 
   (-1)^{|K|-1}\ \big| \phi_K(\graph)\big|$ \quad and \quad 
  $\mycount_3\ \deff\ \mycount_1 - \mycount_2$.
\end{itemize}

Note that $\mycount_3=|\phi_c(\graph)|$ is the desired size of
the query result. Therefore, the $\COUNT$ routine can answer in
time $\bigOh(1)$ by just outputting the number $\mycount_3$.

It remains to show how these values can be initialised and updated
during updates of $\graph$.
The initialisation for the empty graph initialises all the
values to 0.
In the $\UPDATE$ routine, the values for $|\relT_i^\graph|$ and
$\mycount_1$ can be updated in a straightforward way (using time
$\bigOh(c)$).
For each $K\subseteq J$, the update of $|\phi_K(\graph)|$ is
provided within time $\Odc$ by the following
Claim~\ref{claim:counting_positive_formulas}.
\begin{claim}\label{claim:counting_positive_formulas}
  For every $K \subseteq J$, the cardinality
  $\setsize{\phi_K(\graph)}$ of a $\schema_c$-db $\graph$ of degree at
  most $d$ can be updated within time $\Odc$ after
  $\Odc\cdot\card{\graphStart}$ preprocessing time. 
\end{claim}
\begin{proof}
Consider the directed graph $H\deff(V,K)$ with vertex set $V\deff[c]$
and edge set 
$K$.
Decompose the Gaifman graph of $H$ into its connected components. 
Let $V_1,\ldots,V_s$ be the connected components (for a suitable
$s\leq c$).
For each $i\in[s]$ let $H_i\deff H[V_i]$ be the induced subgraph of
$H$ on $V_i$. We write $K_i$ to denote the set of edges of $H_i$.
For every $i\in[s]$ let $\ell_i=|V_i|$, and let $t(i,1)<t(i,2)<\cdots <
t(i,\ell_i)$ be the ordered list of the vertices in $V_i$.
Consider the query
\begin{equation}\label{eq:2}
  \phi_{K_i}(z_{t(i,1)},\ldots,z_{t(i,\ell_i)})
  \quad \deff \quad 
  \bigwedge_{j\in V_i} \relT_j(z_j)
    \ \land \bigwedge_{(j,j')\in K_i} \relE(z_{j},z_{j'}).
\end{equation}
Note that $\phi_K$ is the conjunction of the formulas $\phi_{K_i}$ for
all $i\in[s]$.
Since the variables of the formulas $\phi_{K_i}$ for $i\in[s]$ are
pairwise disjoint, we have \
$\phi_K(\graph) = \phi_{K_1}(\graph) \times \cdots
\times \phi_{K_s}(\graph)$ (modulo permutations of the tuples), and
thus \ $|\phi_K(\graph)|=\prod_{i\in[s]}|\phi_{K_i}(\graph)|$.

For each $i\in[s]$, the value $|\phi_{K_i}(\graph)|$ can be computed as
follows.
For every $v\in \adom{\graph}$ we consider the set
  $S_i^v\deff \setc{(w_{t(i,1)},\ldots,w_{t(i,\ell_i)})\in
    \phi_{K_i}(\graph)}{w_{t(i,1)}=v}$.
Since the Gaifman graph of $H_i$ is connected and has $\ell_i$ nodes,
it follows that each 
component of every tuple
  in $S_i^v$ is contained in the $\ell_i$-neighbourhood of $v$ in
  $\graph$, and this neighbourhood contains 
   at most $d^{\ell_i+1}$ elements.
  Therefore, $\setsize{S_i^v}\leq (d^{\ell_i+1})^{\ell_i}$, and  using
  breadth-first search starting from $v$, the set $S_i^v$ can be
  computed in time \Odc. 
  Note that  $\phi_{K_i}(\graph)$ is 
the disjoint union of the sets $S_i^v$ for all $v\in\adom{\graph}$.
Therefore, $\setsize{\phi_{K_i}(\graph)}=\sum_{v\in\adom{\graph}}\setsize{S_i^v}$.

In our data structure we store for every $i\in[s]$ and every $v\in
\adom{\graph}$ the number $\mycount_{i,v} = \setsize{S_i^v}$.
Moreover, for every $i\in[s]$ we store the sum $\mycount_i = \sum_{v\in
  \adom{\graph}}\mycount_{i,v}=\setsize{\phi_{\relK_i}(\graph)}$.

The initialisation for the empty $\schema_c$-db $\graphStart$ sets all
these values to 0. Whenever 
the colour of a vertex of $\graph$ is updated
or an edge is inserted or deleted, we update all affected numbers accordingly.
Note that a number $\mycount_{i,v}$ changes only if $v$ is in the
$c$-neighbourhood around the updated edge or vertex
in the graph $\graph$. 
Hence, for at most $2d^{c+1}$ vertices $v$, the numbers
$\mycount_{i,v}$ are affected by an update, and
each of them can be updated in time \Odc. 
Moreover, for each $i\in[s]$, the sum $\mycount_i$ can be updated in
time $\bigoh(d^{c+1})$ by subtracting the old value of $\mycount_{i,v}$
and adding the new value of $\mycount_{i,v}$ for each of the at most
$2d^{c+1}$ relevant vertices $v$.
Finally, it takes time $\bigOh(c)$ to compute the updated value
$|\phi_K(\graph)|=\prod_{i\in[s]}\mycount_i$. 
The overall time used to produce the update is \Odc.
\end{proof}

Once we have available the updated numbers $|\phi_K(\graph)|$ for all
$K\subseteq J$, the value $\mycount_2$ can be computed in time 
$\bigOh(|2^{J}|)\leq 2^{\bigOh(c^2)}$. And $\mycount_3$ is then
obtained in time $\bigOh(1)$.
Altogether, performing the $\UPDATE$ routine takes time at most
$\updatetimeCountPhi$.
The $\PREPROCESS$ routine initialises all values for the
empty graph and then uses $|\graphStart|$ update steps to insert all
the tuples of $\graphStart$ into the data structure.
This completes the proof of Lemma~\ref{lem:countPhic}.
\end{proof}

\makeatletter{}%
\section{Enumerating Results of \fomodtext Queries Under Updates}
\label{section:enumerating}

\newcommand{\smallc}{c}%
\newcommand{\setJ}{J}
\renewcommand{\indexSetSmall}{{I}}
\renewcommand{\colourG}[1]{\ensuremath{\relT_{#1}^{\structG}}}
\newcommand{\tupleSetSmallnoIndex}{\tupleSetSmall{}}

\newcommand{\delaysmall}{\ensuremath{c^2}}
\newcommand{\delayEnumSphere}{\ensuremath{k^3}} %
\newcommand{\delaylarge}{\ensuremath{c^3d}}
\renewcommand{\rtdelayd}{\delaylarge}
\newcommand{\updatetimeEnumPhi}{d^{\poly(c)}}
\newcommand{\updatetimeEnumSphere}{\ensuremath{2^{\poly(k)\size{\schema}d^{2r+2}}}}

In this section we prove (and afterwards, improve) the following theorem.

\begin{theorem}\label{thm:enumeration}
There is a dynamic algorithm that receives a schema $\schema$, a
degree bound $d\geq 2$, a $k$-ary $\FOmod[\schema]$-query $\phi(\ov{x})$ (for
some $k\in\NN$), and a
$\schema$-db $\DBstart$ of degree $\leq d$, and computes within 
$\preprocessingtime= f(\phi,d)\cdot\size{\DBstart}$
preprocessing time a data structure that can be updated in time
$\updatetime= f(\phi,d)$ and allows to
enumerate $\phi(\DB)$ with  $d^{2^{\bigOh(\size{\phi})}}$ delay.
\\
The function $f(\phi,d)$ is of the form
$2^{d^{2^{\bigOh(\size{\phi})}}}$.
\end{theorem}

The theorem follows immediately from
Lemma~\ref{lem:translationtograph}\eqref{eq:lem:db2graph:part2} 
and the following dynamic enumeration
algorithm for the query $\phi_c(\ov{z})$.

\begin{lemma}\label{lem:enumPhicWeak}
There is a dynamic algorithm that receives a number $c\geq 1$,
a degree bound $d\geq 2$, and a $\schema_c$-db $\graphStart$ of degree
$\leq d$, and computes
within $\preprocessingtime = \updatetimeEnumPhi \cdot
\card{\graphStart}$ preprocessing 
time a data structure
that can be updated in time $\updatetimeEnumPhi$  and allows to
enumerate the query result $\phi_c(\graph)$ with
$\bigOh(\delaylarge)$ delay.
\end{lemma}
\begin{proof}
For a $\schema_c$-db $\graph$ and a vertex $v\in \adom{\graph}$ we let
$\Ngraph(v)$ be the set of all neighbours of $v$ 
in $\graph$. I.e.,
$\Ngraph(v)$ is the set of all $w\in\adom{\graph}$ such that
$(v,w)$ or $(w,v)$ belongs to $\relE^{\graph}$.

The underlying idea of the enumeration procedure is the following
greedy strategy.  We cycle through all vertices 
$\varu_1\in\colourG{1}$,
$\varu_2\in\colourG{2}\setminus \Ngraph(\varu_1)$,
$\varu_3\in\colourG{3}\setminus \bigl(\Ngraph(\varu_1)\cup
\Ngraph(\varu_2)\bigr)$, \ldots,
$\varu_c\in\colourG{c}\setminus \bigcup_{i\leq c-1}\Ngraph(\varu_i)$ and
output $(\varu_1,\ldots,\varu_c)$.
This strategy does not yet lead to a constant delay enumeration, as there might be vertex tuples
$(\varu_1,\ldots,\varu_i)$ (for $i<c$) that do extend to an output tuple
$(\varu_1,\ldots,\varu_c)$, but where many possible extensions are
checked before this output tuple is encountered.  
We now show
how to overcome this problem and describe an enumeration
procedure with $\bigoh(\delaylarge)$ delay and update time
$\updatetimeEnumPhi$. 
  
Note that for every $\setJ\subseteq[c]$ we have
$\Setsize{\bigcup_{j\in \setJ}\Ngraph(\varu_j)}\leq cd$. Hence, if a set
$\colourG{i}$ contains more than $cd$ elements, we know that \emph{every} considered
tuple has an extension
$\varu_i\in \colourG{i}$ that is not a neighbour of any
vertex in the tuple.
Let $\indexSetSmall \deff \setc{i\in[c]}{\setsize{\colourG{i}}\leq cd}$
be the set of \emph{small} colour classes in $\graph$ and to simplify
the presentation we assume without loss of generality that
$\indexSetSmall=\set{1,\ldots,\ellEnum}$.
In our data structure we store the current index set
$\indexSetSmall$ and the set 
\begin{equation}
  \label{eq:9}
  \tupleSetSmallnoIndex \quad\deff\quad \Setc{\
    (\varu_1,\ldots,\varu_\ellEnum)\; \in\; 
   \colourG{1}\times\cdots\times\colourG{\ellEnum}\ }{\ (\varu_j,\varu_{j'})
   \notin \edgeG, \text{ \;for all }j\neq j'\ }
\end{equation}
of tuples on the small colours. Note that a tuple
$(\varu_1,\ldots,\varu_\ellEnum)\in   \colourG{1}\times\cdots\times\colourG{\ellEnum} $ extends to an output tuple
$(\varu_1,\ldots,\varu_c)\in\phi_c(\graph)$ if and only if
it is contained in $\tupleSetSmallnoIndex$.
We store the current sizes of all colours and this enables us to keep the set
$\indexSetSmall$
of small colours updated.
Moreover, as $\setsize{\tupleSetSmallnoIndex}\leq (cd)^c$,  
we can update the set
$\tupleSetSmallnoIndex$ in time $d^{\poly(c)}$ after every update by a
brute-force approach.
The enumeration procedure is given in Algorithm~\ref{alg:enumeration}.
\begin{algorithm} 
  \caption{Enumeration procedure with delay $\bigoh(c^3d)$}\label{alg:enumeration}
\begin{algorithmic}[1]
\ForAll{$(\varu_1,\ldots,\varu_\ellEnum)\in \tupleSetSmallnoIndex$}
\textsc{Enum}$(\varu_1,\ldots,\varu_\ellEnum)$.
\EndFor
\State{Output the end-of-enumeration message $\EOE$.}
\State{}
  \Function{Enum}{$\varu_1,\ldots,\varu_i$}
   \If{$i = \smallc$} output the tuple $(\varu_1,\ldots,\varu_\smallc)$.
   \Else
   \ForAll{$\varu_{i+1}\in\colourG{i+1}$} \label{line:forloop}
     \If{$\varu_{i+1}\notin \bigcup^i_{j=1}
       \Ngraph(\varu_j)$} \label{line:neighbourhoodtest}
     \textsc{Enum}$(\varu_1,\ldots,\varu_i,\varu_{i+1})$. \label{line:recursionenum}
     \EndIf
   \EndFor
   \EndIf
   \EndFunction
\end{algorithmic}
\end{algorithm}

It is straightforward to see that this procedure enumerates
$\phi_c(\graph)$. Let us analyse the delay. Since for all $i>\ellEnum$ we
have $\Setsize{\colourG{i}}>cd$, it follows that every call of
\textsc{Enum}$(\varu_1,\ldots,\varu_i)$ leads to at least one recursive
call of \textsc{Enum}$(\varu_1,\ldots,\varu_i,\varu_{i+1})$.
Furthermore, there are at most $cd$ iterations of the loop in line
\ref{line:forloop} that do \emph{not} lead to a recursive call. As every test
in line~\ref{line:neighbourhoodtest} can be done in time $\bigoh(c)$,
it follows that the time spans until the first recursive call, between
the calls, and after the last call are bounded by $\bigoh(c^2d)$.  As
the recursion depth is $c$, the overall delay between two output
tuples is bounded by $\bigoh(c^3d)$.
\end{proof}

By using similar techniques as in
\cite{DBLP:conf/pods/DurandSS14},
we obtain the following improved version of
Lemma~\ref{lem:enumPhicWeak} where the delay is independent of the
degree bound $d$.

\begin{lemma}\label{lem:enumPhic}
There is a dynamic algorithm that receives a number $c\geq 1$,
a degree bound $d\geq 2$, and a $\schema_c$-db $\graphStart$ of degree
$\leq d$, and computes
within $\preprocessingtime = \updatetimeEnumPhi \cdot
\card{\graphStart}$ preprocessing 
time a data structure
that can be updated in time $\updatetimeEnumPhi$  and allows to
enumerate the query result $\phi_c(\graph)$ with
$\bigOh(\delaysmall)$ delay.
\end{lemma}

Before proving Lemma~\ref{lem:enumPhic}, let us first point out that
Lemma~\ref{lem:enumPhic} in combination with
Lemma~\ref{lem:translationtograph}\eqref{eq:lem:db2graph:part2} directly improves the
delay in Theorem~\ref{thm:enumeration} from $d^{2^{\bigOh(\size{\phi})}}$ to 
$\bigOh(\delayEnumSphere)$, immediately leading
to the following theorem.

\begin{theorem}\label{thm:enumeration-improved}
There is a dynamic algorithm that receives a schema $\schema$, a
degree bound $d\geq 2$, a $k$-ary $\FOmod[\schema]$-query $\phi(\ov{x})$ (for
some $k\in\NN$), and a
$\schema$-db $\DBstart$ of degree $\leq d$, and computes within 
$\preprocessingtime= f(\phi,d)\cdot\size{\DBstart}$
preprocessing time a data structure that can be updated in time
$\updatetime= f(\phi,d)$ and allows to
enumerate $\phi(\DB)$ with  $\bigOh(\delayEnumSphere)$ delay.
\\
The function $f(\phi,d)$ is of the form 
$2^{d^{2^{\bigOh(\size{\phi})}}}$.
\end{theorem}
The rest of the section is devoted to the proof of Lemma~\ref{lem:enumPhic}.
\makeatletter{}%

\begin{proof}[Proof of Lemma~\ref{lem:enumPhic}]

Consider Algorithm~\ref{alg:enumeration}, which enumerates $\phi_c(\graph)$ with
$\bigOh(\delaylarge)$ delay.
To enumerate the tuples with only
$\bigoh(\delaysmall)$ delay, we replace the loop in lines
\ref{line:forloop}--\ref{line:recursionenum} by a precomputed
``skip'' function that allows to iterate through all elements in
$\colourG{i+1}\setminus \bigcup^i_{j=1} \Ngraph(\varu_j)$ with $\bigoh(c)$
delay.

For every $i \in [c]$ we store all elements of $\colourG{i}$ in a
doubly linked list and let $\void$ be an auxiliary element that
appears at the end of the list. We let $\firstelement_i$ be the first element
in the list and $\successor_i(\varu)$ the successor of
$\varu\in\colourG{i}$.
We denote by $\leq^i$ the linear order induced by this list. 
We let $\relEtilde^\graph$ be the symmetric closure of $\relE^\graph$,
i.e., $\relEtilde^\graph=\relE^\graph\cup\setc{(v,u)}{(u,v)\in\relE^\graph}$.
For every $i\in[c]$ we define the function 
\[
 \skipp{i}(\vary,\setV)\quad :=\quad 
 \min\left\{ 
   \varz\in\colourG{i}\cup\set{\void}
   \ : \ 
   \vary \leq^i \varz 
   \text{ \ and \ } 
   \text{for all } \varv \in \setV, \  
   (\varv,\varz) \notin \relEtilde^\graph 
 \right\},
\]
which assigns to every $\setV\subseteq \adom{\graph}$ with
$\setsize{\setV}\leq c{-}1$, and every $\vary\in\colourG{i}$ the next node that
is not adjacent to any vertex in $\setV$.

Using these functions, our improved enumeration
algorithm is given in Algorithm~\ref{alg:enumerationconstNEW}.
Below, we show that we can access the values
$\skipp{i}(\vary,\setV)$ in time $\bigOh(c)$. By the same analysis as given
in the proof of Lemma~\ref{lem:enumPhicWeak} it then follows that 
Algorithm~\ref{alg:enumerationconstNEW} enumerates $\phi_c(\graph)$ with $\bigOh(\smallc^2)$ delay.

\begin{algorithm}[h!] 
  \caption{Enumeration procedure with delay $\bigOh(\smallc^2)$}\label{alg:enumerationconstNEW}
\begin{algorithmic}[1]
\ForAll{$(\varu_1,\ldots,\varu_\ellEnum)\in \tupleSetSmallnoIndex$}\label{line:forelementsS}
\State \textsc{Enum}$(\varu_1,\ldots,\varu_\ellEnum)$.
\EndFor
\State{Output the end-of-enumeration message $\EOE$.}
\State{}
  \Function{Enum}{$\varu_1,\ldots,\varu_i$}
  \If{$i = \smallc$}
   \State output the tuple $(\varu_1,\ldots,\varu_\smallc)$.
   \Else
   \State $\vary \gets \skipp{i+1}(\firstelement_{i+1},\set{\varu_1,\ldots,\varu_i})$  \label{line:ysmallest} 
   \While{$\vary \neq \void$} \label{line:loopwhile}
   \State \textsc{Enum}$(\varu_1,\ldots,\varu_{i},\vary)$. 
   \State $\vary \gets \skipp{i+1}(\successor_{i+1}(\vary),\set{\varu_1,\ldots,\varu_i})$. \label{line:???}
   \EndWhile
   \EndIf \label{line:endlastif}
   \EndFunction
\end{algorithmic}
\end{algorithm}

What remains to show is that we can access the values
$\skipp{i}(\vary,\setV)$ for all $i,\vary,V$ in time $\bigOh(c)$ and
maintain them with $\updatetimeEnumPhi$ update time.
At first sight, this is not clear at all, because the domain of
$\skipp{i}$ has size $\Omega(|\adom{\graph}|^c)$.
In what follows, we show that for every $\vary$, the number of distinct
values that $\skipp{i}(\vary,\setV)$ can take is bounded by
$\updatetimeEnumPhi$, and that we can store them in a look-up
table with update time $\updatetimeEnumPhi$.

To illustrate the main idea, let us start with a simple example. 
We want to enumerate $\phi_4$ on a coloured graph $\graphH$ with four
vertex colours blue, red, yellow, and 
green (in this order) and analyse the call of 
\textsc{Enum}$(b,r,y)$, which is supposed to enumerate all green nodes
$g_i$ that are not adjacent to any of the nodes $b$, $r$, and $y$.
The relevant part of $\graphH$ is depicted in Figure \ref{fig:graphh}.

\begin{figure}[htbp]
 \begin{tikzpicture}
 \node[draw=green,fill=green!20,circle] (g1) at (1,1) {$g_1$};
 \node[draw=green,fill=green!20,circle] (g2) at (2,1) {$g_2$};
 \node[draw=green,fill=green!20,circle] (g3) at (3,1) {$g_3$};
 \node[draw=green,fill=green!20,circle] (g4) at (4,1) {$g_4$};
 \node[draw=green,fill=green!20,circle] (g5) at (5,1) {$g_5$};
 \node[draw=green,fill=green!20,circle] (g6) at (6,1) {$g_6$};
 \node[draw=green,fill=green!20,circle] (g6) at (6,1) {$g_6$};
 \node[draw=blue,fill=blue!20,circle] (b) at (2,-1) {$b$};
 \node[draw=red,fill=red!20,circle] (r) at (4,-1) {$r$};
 \node[draw=yellow,fill=yellow!20,circle] (y) at (6,-1) {$\vary$};
 \draw[thick] (b) -- (g1);
 \draw[thick] (b) -- (g2);
 \draw[thick] (b) -- (g3);
 \draw[thick] (r) -- (g4);
 \draw[thick] (y) -- (g6);
\end{tikzpicture}
\centering
\caption{Illustration of the relevant part of graph $\graphH$}\label{fig:graphh}
\end{figure}
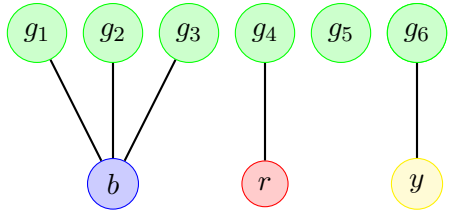

The enumeration procedure starts by considering the first element $g_1$ in the
list of green vertices, but the first element in the actual
output is $g_5=\skipp{i}(g_1,\set{b,r,y})$. 
Therefore, we have to skip the irrelevant vertices $g_1,\ldots,g_4$.

To do this, we want to know the neighbours of the vertices
that we skip ($b$ and $r$ in our example) when looking at $g_1$. For
this purpose, we define inductively new sorts of edges
$\edgeE{4}{1}\subseteq\edgeE{4}{2}\subseteq \cdots$ that connect
green vertices $g_i$ with $\relEtilde$-neighbours of skipped vertices.
In our example, we first have to skip $g_1$, because it is $\relEtilde$-connected to $b$ and
we indicate this by letting $\edgeE{4}{1}$ be the set of tuples $(g_i,\varv)\in \relEtilde^\graphH$ (see Figure~\ref{fig:graphh-eone}).

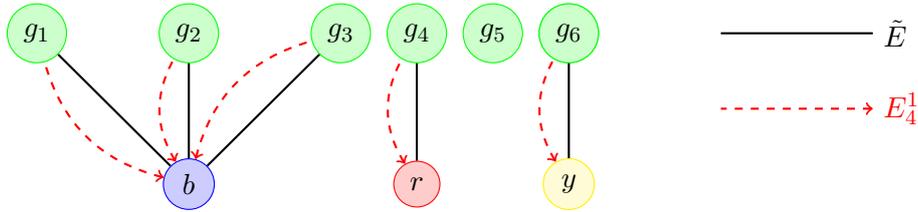
\begin{figure}[htbp]
 \begin{tikzpicture}
 \node[draw=green,fill=green!20,circle] (g1) at (-1,1) {$g_1$};
 \node[draw=green,fill=green!20,circle] (g2) at (1,1) {$g_2$};
 \node[draw=green,fill=green!20,circle] (g3) at (3,1) {$g_3$};
 \node[draw=green,fill=green!20,circle] (g4) at (4,1) {$g_4$};
 \node[draw=green,fill=green!20,circle] (g5) at (5,1) {$g_5$};
 \node[draw=green,fill=green!20,circle] (g6) at (6,1) {$g_6$};
 \node[draw=green,fill=green!20,circle] (g6) at (6,1) {$g_6$};
 \node[draw=blue,fill=blue!20,circle] (b) at (1,-1) {$b$};
 \node[draw=red,fill=red!20,circle] (r) at (4,-1) {$r$};
 \node[draw=yellow,fill=yellow!20,circle] (y) at (6,-1) {$\vary$};
 \draw[thick] (b) -- (g1);
 \draw[thick] (b) -- (g2);
 \draw[thick] (b) -- (g3);
 \draw[thick] (r) -- (g4);
 \draw[thick] (y) -- (g6);
 \tikzstyle{eoneedges}=[red,thick,bend left,bend angle=45,<-,dashed]
 \draw[eoneedges] (b) to (g1);
 \draw[eoneedges] (b) to (g2);
 \draw[eoneedges] (b) to (g3);
 \draw[eoneedges] (r) to (g4);
 \draw[eoneedges] (y) to (g6);
 \draw[thick] (8,1) -- (10,1) node[right] {$\relEtilde$};
 \draw[eoneedges] (10,0) node[right] {$\edgeE{4}{1}$}-- (8,0) ;
\end{tikzpicture}
\centering
\caption{$\relEtilde$-edges and $\edgeE{4}{1}$-edges in our example}\label{fig:graphh-eone}
\end{figure}
After realising that even more vertices ($g_2$ and $g_3$) are excluded
by $b$, the next try would be $g_4$.
However, this vertex is excluded by its $\relEtilde$-neighbour $r$, so we have to
take $r$ into account when computing the skip value for $g_1$ and
indicate this by the $\edgeE{4}{2}$-edge $(g_1,r)$ (see Figure~\ref{fig:graphheoneetwo}).
This immediately leads to an inductive definition: $\edgeE{4}{2}$
contains all pairs of vertices that are already in $\edgeE{4}{1}$ or
connected by a path as shown in Figure~\ref{fig:pathetwo}.

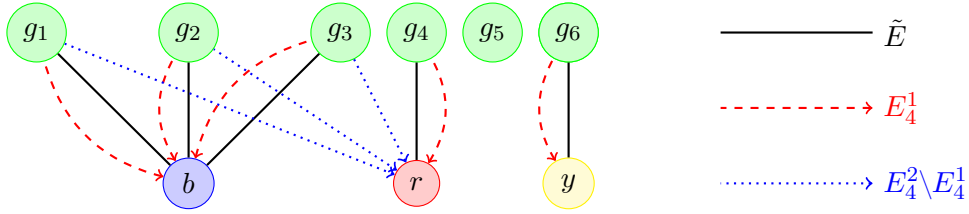
\begin{figure}[htbp]
 \begin{tikzpicture}
 \node[draw=green,fill=green!20,circle] (g1) at (-1,1) {$g_1$};
 \node[draw=green,fill=green!20,circle] (g2) at (1,1) {$g_2$};
 \node[draw=green,fill=green!20,circle] (g3) at (3,1) {$g_3$};
 \node[draw=green,fill=green!20,circle] (g4) at (4,1) {$g_4$};
 \node[draw=green,fill=green!20,circle] (g5) at (5,1) {$g_5$};
 \node[draw=green,fill=green!20,circle] (g6) at (6,1) {$g_6$};
 \node[draw=green,fill=green!20,circle] (g6) at (6,1) {$g_6$};
 \node[draw=blue,fill=blue!20,circle] (b) at (1,-1) {$b$};
 \node[draw=red,fill=red!20,circle] (r) at (4,-1) {$r$};
 \node[draw=yellow,fill=yellow!20,circle] (y) at (6,-1) {$\vary$};
 \draw[thick] (b) -- (g1);
 \draw[thick] (b) -- (g2);
 \draw[thick] (b) -- (g3);
 \draw[thick] (r) -- (g4);
 \draw[thick] (y) -- (g6);
 \tikzstyle{eoneedges}=[red,thick,bend left,bend angle=45,<-,dashed]
 \tikzstyle{eoneedgesr}=[red,thick,bend right,bend angle=45,<-,dashed]
 \draw[eoneedges] (b) to (g1);
 \draw[eoneedges] (b) to (g2);
 \draw[eoneedges] (b) to (g3);
 \draw[eoneedgesr] (r) to (g4);
 \draw[eoneedges] (y) to (g6);
 \tikzstyle{etwoedges}=[blue,thick,<-,dotted]
 \draw[etwoedges] (r) to (g1);
 \draw[etwoedges] (r) to (g2);
 \draw[etwoedges] (r) to (g3); 
 \draw[thick] (8,1) -- (10,1) node[right] {$\relEtilde$};
 \draw[eoneedges] (10,0) node[right] {$\edgeE{4}{1}$}-- (8,0) ;
 \draw[etwoedges] (10,-1) node[right] {$\edgeE{4}{2} \setminus \edgeE{4}{1}$}-- (8,-1) ;
\end{tikzpicture}
\centering
\caption{$\relEtilde$-edges, $\edgeE{4}{1}$-edges and $\edgeE{4}{2}$-edges in our example}\label{fig:graphheoneetwo}
\end{figure}

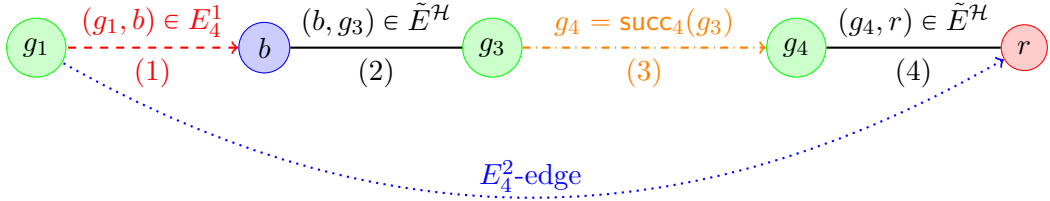
\begin{figure}[htbp]
\begin{tikzpicture}
 \node[draw=green,fill=green!20,circle] (g1) at (0,0) {$g_1$};
 \node[draw=blue,fill=blue!20,circle] (b) at (3,0) {$b$};
 \node[draw=green,fill=green!20,circle] (g3) at (6,0) {$g_3$};
 \node[draw=green,fill=green!20,circle] (g4) at (10,0) {$g_4$};
 \node[draw=red,fill=red!20,circle] (r) at (13,0) {$r$};
 \tikzstyle{eoneedges}=[red,thick,<-,dashed]
 \draw[eoneedges] (b) -- (g1) node[midway,above] {$(g_1,b) \in \edgeE{4}{1}$} node[midway,below] {$(1)$};
 \draw[thick] (b) -- (g3) node[midway,above] {$(b,g_3) \in \relEtilde^\graphH$} node[midway,below] {$(2)$};
 \draw[orange,dashdotted,->,thick] (g3) -- (g4) node[midway,above] {$g_4 = \successor_4(g_3)$} node[midway,below] {$(3)$};
 \draw[thick] (r) -- (g4) node[midway,above] {$(g_4,r) \in \relEtilde^\graphH$} node[midway,below] {$(4)$}; 
 \tikzstyle{etwoedges}=[blue,thick,<-,dotted]
 \draw[etwoedges,bend left, bend angle=45] (r) to (g1) ;
 \node[blue] at (6.5,-1.7) {$\edgeE{4}{2}$-edge};
\end{tikzpicture}
\centering
 \caption{Introducing an $\edgeE{4}{2}$-edge between $g_1$ and $r$}\label{fig:pathetwo}
\end{figure}

The idea outlined above can be formalised as follows.
For $i,j \in [c]$, we define inductively the
auxiliary edge sets $\edgeE{i}{j}$:
\begin{align*}
 \edgeE{i}{1} &\ \ := \ \ \setc{\,(\vary,\varu)}{\vary\in
                              C_i^\graph \text{ and } (\vary,\varu)
                              \in \relEtilde^\graph\,} \quad\text{and}
                            \\
 \edgeE{i}{j+1} &\ \ := \ \ \edgeE{i}{j}\,\cup\,
     \Setc{\,(\vary,\varu)}{
              \text{there are } \varv, \varz \text{ with }
              (\vary,\varv)\in\edgeE{i}{j}, \
              (\varv,\varz)\in \relEtilde^\graph, \
              (\successor_i(\varz),\varu)\in\relEtilde^\graph
      \,}
\end{align*}
Now we define for every $\vary\in C_i^\graph$ the set 
\[
  \setS_i^\vary \ \ := \ \ \setc{\,\varu}{(y,u)\in \edgeE{i}{c}\,}\,.
\]
Note that $\setsize{\setS_i^\vary}= \bigOh(d^{2c})$.
The following claim states that $\setS_i^\vary$ are the only 
vertices we need to take into account when computing
$\skipp{i}(\vary,\setV)$.

\begin{claim}\label{claim:succinct_skip}
For all $i\leq c$, $\vary\in C_i^\graph\cup\set{\void}$, and
$\setV\subseteq\adom{\graph}$ with $\setsize{V}\leq c{-}1$ it holds that
\begin{equation}
\skipp{i}(\vary,\setV)
\ \ = \ \
\skipp{i}(\vary,\setV\cap\setS_i^\vary)\,.  
\end{equation}
\end{claim}

\begin{proof}
The proof is identical to the proof of Claim~1 in \cite{DBLP:conf/pods/DurandSS14}.
For the reader's convenience, we include a proof here.

If $\vary=\void$, the lemma is trivial. Hence assume that
$\vary\neq\void$ and let $\varz := \skipp{i}(\vary,\setV\cap\setS_i^\vary)$. By definition we
have $\vary\leq^i\varz \leq^i \skipp{i}(\vary,\setV)$ and therefore we have to
show $\varz \geq^i \skipp{i}(\vary,\setV)$, which holds if and only if
$(\varu,\varz)\notin \relEtilde^\graph$ for all $\varu\in \setV\setminus
\setS_i^\vary$.
If $\varz=\vary$, the claim clearly 
holds as all $\relEtilde^\graph$-neighbours of $\vary$ are
contained in $\setS_i^\vary$.
Hence we have $\varz >^i \vary$ and let $\varz'\geq^i\vary$ be the
predecessor of $\varz$, i.e., $\varz=\successor_i(\varz')$.
Now assume for contradiction that there is an $\varu\in \setV\setminus
\setS_i^\vary$ such that (${\ast}$)~$(\varu,\varz)\in \relEtilde^\graph$.
Note that since $\varz'<^i\varz=\skipp{i}(\vary,\setV\cap\setS_i^\vary)$, there is
a $\varv\in \setV\cap\setS_i^\vary$ such that (${\ast}{\ast}$)~$(\varv,\varz')\in \relEtilde^\graph$.
In the following we show that 
(${\ast}{\ast}{\ast}$)~$(\vary,\varv)\in\edgeE{i}{c-1}$. Note that this finishes the proof of
the claim, as by the definition of
$\edgeE{i}{c}$, the statements (${\ast}$), (${\ast}{\ast}$), and (${\ast}{\ast}{\ast}$) imply
that $\varu\in\setS_i^\vary$, contradicting the assumption that $\varu\in \setV\setminus
\setS_i^\vary$.

To show that $(\vary,\varv)\in\edgeE{i}{c-1}$, let
\begin{equation}
  \label{eq:8}
  \setV_j \ \ := \ \ \setc{\,\varv'\in \setV}{(\vary,\varv')\in \edgeE{i}{j}\,}
\end{equation}
for all $j\in[c]$.
Note that $\setV_c=\setV\cap \setS_i^{\vary}$.
Furthermore,
if there is a $j<c$ with $\setV_j=\setV_{j+1}$, then we have
\begin{equation}
  \label{eq:10}
  \setV_j \  = \ 
  \setV_{j+1}  \ = \ 
  \cdots \  = \ 
  \setV_{c} \ =  \ \setV\cap\setS_i^{\vary}\,.
\end{equation}
Since $\setsize{V}\leq c{-}1$ and $\varu\in \setV\setminus
\setS_i^\vary$, we have $\setsize{\setV\cap\setS_i^{\vary}}\leq c-2$.
In particular, it holds that $\setV_{c-1}=\setV\cap\setS_i^{\vary}$.
Since $\varv\in\setV\cap\setS_i^{\vary}$, it holds that
$\varv\in\setV_{c-1}$ and thus $(\vary,\varv)\in\edgeE{i}{c-1}$.
\end{proof}

In our dynamic algorithm we maintain an array that allows random
access to the values $\skipp{i}(\vary,\setS')$ for all $\vary\in
C_i^\graph$ and all $\setS'\subseteq \setS_i^\vary$ of size at most $c{-}1$.
By Claim~\ref{claim:succinct_skip} we can then compute
$\skipp{i}(\vary,V)$ by first computing $\setS'=\setV\cap\setS_i^\vary$ and
then looking up $\skipp{i}(\vary,\setS')$. This can be done in time $\bigOh(c)$. 
The next claim states that we can efficiently maintain the sets $\setS_i^\vary$.

\begin{claim}\label{claim:updateS}
  There is a data structure that
 \begin{enumerate}
 \item stores the elements from the sets $\setS_i^{\vary}$ and
 all subsets $\setS'\subseteq\setS_i^{\vary}$ of cardinality at most
 $c{-}1$,
 \item 
   allows to test membership in these sets in time $\bigOh(1)$, and
 \item can be updated in time $d^{\poly(c)}$ after every update of
     the form $\Insert\;C_i(v)$, $\Delete\;C_i(v)$, $\Insert\;E(u,v)$,
     and $\Delete\;E(u,v)$.
 \end{enumerate}
\end{claim}

\begin{proof}
 Note that $u\in \setS_i^{\vary}$ $\iff$ $(y,u)\in \edgeE{i}{c}$.
 We store the edge sets $\edgeE{i}{j}$ for all $i,j\in[c]$ in adjacency lists and
 additionally maintain arrays to allow constant-time access to the list
 entries. This allows us to store a list of elements from
 $\setS_i^{\vary}$ and access the elements in $\setS_i^{\vary}$ in
 constant time. 
 Moreover, as the size of $\setS_i^{\vary}$ is bounded
 by $\bigOh(d^{2c})$, 
 the number of subsets $\setS'\subseteq\setS_i^{\vary}$ of cardinality at most
 $c{-}1$ is bounded by $\bigOh(d^{3c})$. Consequently, 
 we can provide constant-time access to all these subsets $\setS'$.

 On every insertion or deletion of an edge in $E^{\graph}$, as well as every insertion or
 deletion of a vertex in $C_i^\graph$, at most $\bigOh(d)$ pairs in the
 relation $\edgeE{i}{1}$ change and the relation can be updated in
 time $\bigOh(d)$.
 Afterwards we update the edge sets $\edgeE{i}{j}$ according
 to their inductive definition. To do this efficiently, we use a
 breadth-first search starting from $\varu$ and $\varv$, for every
 tuple $(\varu,\varv)$ that has changed in relation $\edgeE{i}{1}$,  up to depth
 $3c$ to identify the relevant nodes that are affected by the change.
 By using the adjacency lists, this can
 be done in time $d^{\poly(c)}$ as the degree of the edge sets is
 bounded by $d^{\poly(c)}$.
 We leave the details to the reader.
 \end{proof}

 In our data structure we store the values
 $\skipp{i}(\vary,\setS')$
 for every $i \in [\smallc]$, $\vary \in \colourG{i}$ and for all sets
 $\setS'\subseteq\setS_i^{\vary}$ of cardinality at most $c{-}1$.
 On every insertion or deletion of an edge, we update the sets $\setS_i^{\vary}$ and
 their subsets $\setS'$ of cardinality at most $c{-}1$ and update
 affected values of $\skipp{i}(\vary,\setS')$.
 According to Claim~\ref{claim:updateS} this can be done in time $d^{\poly(\varphi)}$.

 We do the same on updates of the form $\Insert\;C_i(v)$ and
 $\Delete\;C_i(v)$, but have to do some additional work, as $v$ might
 occur in the image of skip-functions.
 Upon $\Insert\;C_i(v)$, we insert $v$ at the beginning of the list
 $C_i$. This ensures that existing skip values will not be
 affected. Afterwards, we compute the set $S_i^v$ and the values
 $\skipp{i}(v,\setS')$ for all $S'\subseteq S_i^v$ of
 cardinality at most $c{-}1$.
 Again, this can be done in time $d^{\poly(\varphi)}$.

 If we receive the update $\Delete\;C_i(v)$, then we have to recompute
 all skip values $\skipp{i}(y,\setS')$ that point to $v$. 
 Note that (since $\graph$ has degree $\leq d$) this is only the case
 for nodes $y\leq^iv$ whose distance from $v$ w.r.t.\ $\successor_i$ is at most $(c{-}1)d$.
 Hence, it suffices to recompute $\skipp{i}(y,\setS')$ for
 at most $(c{-}1)d$ vertices $y$ and all $S'\subseteq  S_i^y$  of cardinality at most $c{-}1$. 
 This can be done in time $d^{\poly(\varphi)}$.
By Claim~\ref{claim:succinct_skip}, all this suffices to access the value
for $\skipp{i}(\vary,\setV)$ in time $\bigOh(c)$. This concludes the
proof of Lemma~\ref{lem:enumPhic}.
\end{proof}

\makeatletter{}%
\section{Conclusion}\label{sec:conclusion}

Our main results show that in the dynamic setting (i.e., allowing
database updates), the results of
$k$-ary $\FOmod$-queries on bounded degree databases can be
tested and counted in constant time and
enumerated with constant delay,
after linear time preprocessing and with
constant update time. Here, ``constant time'' refers to data
complexity and is of size $\poly(k)$ concerning the delay and
the time for testing and counting. The time for performing a database
update is 3-fold exponential in the size of the query and the degree
bound, and is worst-case optimal.

The starting point of our algorithms is to decompose the given query
into a query in Hanf normal form, using a recent result of
\cite{HKS_Hanf_LICS16}. This normal form is only available for the
setting with a fixed maximum degree bound $d$, i.e., the setting
considered in this paper. 

Recently, Kuske and Schweikardt \cite{KuskeSchweikardt-FOCN}
introduced a new kind of Hanf normal form for a variant of
\emph{first-order logic with counting} that contains and extends Libkin's logic
$\LogicFont{FO(Cnt)}$ \cite{Lib04} and Grohe's logic
$\LogicFont{FO{+}C}$ \cite{Gro13}. As an application it is shown in
\cite{KuskeSchweikardt-FOCN} that the present paper's techniques can
be lifted from $\FOmod$ to 
first-order logic with counting.

An obvious future task is to investigate to
which extent further query evaluation results that are known for the static
setting can be lifted to the dynamic setting. More specifically: Are there efficient
dynamic algorithms for evaluating (i.e., answering, testing, counting, or
enumerating) results of first-order queries on other sparse classes
of databases (e.g.\ planar, bounded treewidth, bounded expansion, nowhere dense)
or databases of low degree, 
lifting the ``static'' results accumulated in
\cite{Kazana.2013,Grohe.2014,DBLP:conf/pods/DurandSS14} to the dynamic
setting?

\bibliography{literature}

\begin{thebibliography}{10}

\bibitem{AHV-Book}
Serge Abiteboul, Richard Hull, and Victor Vianu.
\newblock {\em Foundations of Databases}.
\newblock Addison-Wesley, 1995.
\newblock URL: \url{http://webdam.inria.fr/Alice/}.

\bibitem{BKS_enumeration_PODS17}
Christoph Berkholz, Jens Keppeler, and Nicole Schweikardt.
\newblock Answering conjunctive queries under updates.
\newblock In {\em Proceedings of the 36th {ACM} {SIGMOD-SIGACT-SIGART}
  Symposium on Principles of Database Systems, PODS'17}, 2017.
\newblock To appear, preprint available at
  \url{https://arxiv.org/abs/1702.06370}.

\bibitem{BKS-ICDT17}
Christoph Berkholz, Jens Keppeler, and Nicole Schweikardt.
\newblock Answering {FO$+$MOD} queries under updates on bounded degree
  databases.
\newblock In {\em Proceedings of the 20th International Conference on Database
  Theory, {ICDT'17}, March 21--24, 2017, Venice, Italy}, LIPIcs. Schloss
  Dagstuhl - Leibniz-Zentrum fuer Informatik, 2017.
\newblock To appear.

\bibitem{Cormen.2009}
Thomas~H. Cormen, Charles~E. Leiserson, Ronald~L. Rivest, and Clifford Stein.
\newblock {\em Introduction to Algorithms {(3.}~ed.)}.
\newblock {MIT} Press, 2009.
\newblock URL: \url{http://mitpress.mit.edu/books/introduction-algorithms}.

\bibitem{DurandGrandjean_BoundedDegree}
Arnaud Durand and Etienne Grandjean.
\newblock First-order queries on structures of bounded degree are computable
  with constant delay.
\newblock {\em {ACM} Trans. Comput. Log.}, 8(4), 2007.
\newblock \href {http://dx.doi.org/10.1145/1276920.1276923}
  {\path{doi:10.1145/1276920.1276923}}.

\bibitem{DBLP:conf/pods/DurandSS14}
Arnaud Durand, Nicole Schweikardt, and Luc Segoufin.
\newblock Enumerating answers to first-order queries over databases of low
  degree.
\newblock In {\em Proceedings of the 33rd {ACM} {SIGMOD-SIGACT-SIGART}
  Symposium on Principles of Database Systems, PODS'14, Snowbird, UT, USA, June
  22-27, 2014}, pages 121--131, 2014.
\newblock \href {http://dx.doi.org/10.1145/2594538.2594539}
  {\path{doi:10.1145/2594538.2594539}}.

\bibitem{FrickGrohe_APAL2004}
Markus Frick and Martin Grohe.
\newblock The complexity of first-order and monadic second-order logic
  revisited.
\newblock {\em Ann. Pure Appl. Logic}, 130(1-3):3--31, 2004.
\newblock \href {http://dx.doi.org/10.1016/j.apal.2004.01.007}
  {\path{doi:10.1016/j.apal.2004.01.007}}.

\bibitem{Gro13}
Martin Grohe.
\newblock {\em Descriptive Complexity, Canonisation, and Definable Graph
  Structure Theory}.
\newblock Lecture Notes in Logic. Association for Symbolic Logic in conjunction
  with Cambridge University Press, to appear.
\newblock Preliminary version available at
  \url{https://www.lii.rwth-aachen.de/de/13-mitarbeiter/professoren/39-book-descriptive-complexity.html}.

\bibitem{Grohe.2014}
Martin Grohe, Stephan Kreutzer, and Sebastian Siebertz.
\newblock Deciding first-order properties of nowhere dense graphs.
\newblock In {\em Symposium on Theory of Computing, {STOC} 2014, New York, NY,
  USA, May 31 - June 03, 2014}, pages 89--98, 2014.
\newblock \href {http://dx.doi.org/10.1145/2591796.2591851}
  {\path{doi:10.1145/2591796.2591851}}.

\bibitem{HKS_Hanf_LICS16}
Lucas Heimberg, Dietrich Kuske, and Nicole Schweikardt.
\newblock Hanf normal form for first-order logic with unary counting
  quantifiers.
\newblock In {\em Proceedings of the 31st Annual {ACM/IEEE} Symposium on Logic
  in Computer Science, {LICS} '16, New York, NY, USA, July 5-8, 2016}, pages
  277--286, 2016.
\newblock \href {http://dx.doi.org/10.1145/2933575.2934571}
  {\path{doi:10.1145/2933575.2934571}}.

\bibitem{KazanaSegoufin_BoundedDegree}
Wojciech Kazana and Luc Segoufin.
\newblock First-order query evaluation on structures of bounded degree.
\newblock {\em Logical Methods in Computer Science}, 7(2), 2011.
\newblock \href {http://dx.doi.org/10.2168/LMCS-7(2:20)2011}
  {\path{doi:10.2168/LMCS-7(2:20)2011}}.

\bibitem{Kazana.2013}
Wojciech Kazana and Luc Segoufin.
\newblock Enumeration of first-order queries on classes of structures with
  bounded expansion.
\newblock In {\em Proceedings of the 32nd {ACM} {SIGMOD-SIGACT-SIGART}
  Symposium on Principles of Database Systems, {PODS} 2013, New York, NY,
  {USA}, June 22--27, 2013}, pages 297--308, 2013.
\newblock \href {http://dx.doi.org/10.1145/2463664.2463667}
  {\path{doi:10.1145/2463664.2463667}}.

\bibitem{KuskeSchweikardt-FOCN}
Dietrich Kuske and Nicole Schweikardt.
\newblock First-order logic with counting: At least, \emph{weak} {H}anf normal
  forms always exist and can be computed!
\newblock Manuscript, 2017.

\bibitem{Lib04}
Leonid Libkin.
\newblock {\em Elements of Finite Model Theory}.
\newblock Texts in Theoretical Computer Science. An {EATCS} Series. Springer,
  2004.
\newblock \href {http://dx.doi.org/10.1007/978-3-662-07003-1}
  {\path{doi:10.1007/978-3-662-07003-1}}.

\bibitem{DBLP:journals/jcss/Luks82}
Eugene~M. Luks.
\newblock Isomorphism of graphs of bounded valence can be tested in polynomial
  time.
\newblock {\em J. Comput. Syst. Sci.}, 25(1):42--65, 1982.
\newblock \href {http://dx.doi.org/10.1016/0022-0000(82)90009-5}
  {\path{doi:10.1016/0022-0000(82)90009-5}}.

\bibitem{MoretShapiro}
Bernard M.~E.\ Moret and Henry~D.\ Shapiro.
\newblock {\em Algorithms from P to NP: Volume~1: Design \& Efficiency}.
\newblock Benjamin-Cummings, 1991.

\bibitem{Patnaik.1997}
Sushant Patnaik and Neil Immerman.
\newblock {Dyn-FO:} {A} parallel, dynamic complexity class.
\newblock {\em J. Comput. Syst. Sci.}, 55(2):199--209, 1997.
\newblock \href {http://dx.doi.org/10.1006/jcss.1997.1520}
  {\path{doi:10.1006/jcss.1997.1520}}.

\bibitem{Schwentick.2016}
Thomas Schwentick and Thomas Zeume.
\newblock Dynamic complexity: recent updates.
\newblock {\em {SIGLOG} News}, 3(2):30--52, 2016.
\newblock URL: \url{http://doi.acm.org/10.1145/2948896.2948899}, \href
  {http://dx.doi.org/10.1145/2948896.2948899}
  {\path{doi:10.1145/2948896.2948899}}.

\bibitem{Seese.1996}
Detlef Seese.
\newblock Linear time computable problems and first-order descriptions.
\newblock {\em Mathematical Structures in Computer Science}, 6(6):505--526,
  1996.

\end{thebibliography}

\end{document}